\documentclass[11pt,letterpaper,onecolumn]{article}
\usepackage{fullpage, dsfont, amsthm, amssymb, amsmath,stmaryrd,authblk, mathtools}

\usepackage{graphicx,color,tikz}
\usepackage{relsize}
\usepackage{url}
\usepackage[citecolor=cyan]{hyperref}
\hypersetup{
    colorlinks=true,
    linkcolor=red,
    filecolor=magenta,      
    urlcolor=cyan,
}
\usepackage[capitalise]{cleveref}
\allowdisplaybreaks

\usepackage{michaels_macros}

\usepackage{soul}

\title{Improved quantum backtracking algorithms using effective resistance estimates}

\author[1]{Michael Jarret \thanks{mjarret@pitp.ca}}
\affil[1]{\footnotesize\PITP}
\author[1,2]{Kianna Wan \thanks{kianna.wan@gmail.com}}
\affil[2]{\footnotesize Department of Physics and Astronomy, University of Waterloo, Waterloo, Ontario, Canada, N2L 3G1}

\usepackage{enumitem}

\begin{document}
    \maketitle

\renewcommand{\discussion}[1]{}
\newcommand{\half}{\frac{1}{2}}

\begin{abstract}
We investigate quantum backtracking algorithms of the type introduced by Montanaro (arXiv:1509.02374). These algorithms explore trees of unknown structure and in certain settings exponentially outperform their classical counterparts. Some of the previous work focused on obtaining a quantum advantage for trees in which a unique marked vertex is promised to exist. We remove this restriction by recharacterising the problem in terms of the effective resistance of the search space. In this paper, we present a generalisation of one of Montanaro's algorithms to trees containing $k$ marked vertices, where $k$ is not necessarily known \textit{a priori}. 

Our approach involves using amplitude estimation to determine a near-optimal weighting of a diffusion operator, which can then be applied to prepare a superposition state with support only on marked vertices and ancestors thereof. By repeatedly sampling this state and updating the input vertex, a marked vertex is reached in a logarithmic number of steps. The algorithm thereby achieves the conjectured bound of $\widetilde{\mathcal{O}}(\sqrt{TR_{\mathrm{max}}})$ for finding a single marked vertex and $\widetilde{\mathcal{O}}\left(k\sqrt{T R_{\mathrm{max}}}\right)$ for finding all $k$ marked vertices, where $T$ is an upper bound on the tree size and $R_{\mathrm{max}}$ is the maximum effective resistance encountered by the algorithm. This constitutes a speedup over Montanaro's original procedure in both the case of finding one and the case of finding multiple marked vertices in an arbitrary tree. 
\end{abstract}
\maketitle
\newpage
\section{Introduction}


In this paper, we construct a quantum algorithm for finding marked vertices in a rooted tree of unknown structure. Search trees are a natural data structure for many computational problems, notably constraint satisfaction problems (CSPs). Consider a CSP defined on a finite domain $[d-1] = \{0,1,\dots,d-1\}$ by a predicate $P:[d-1]^n \longrightarrow \{\mathrm{true},\mathrm{false}\}$. The objective of the decision problem is to determine whether there is at least one assignment $x \in [d-1]^n$ such that $P(x)$ evaluates to true. In many cases, we would like an algorithm capable not only of deciding whether a satisfying assignment exists, but also of returning one or more such assignments.

Backtracking is a prevailing technique for solving CSPs, and in Ref.\ \cite{montanaro2015quantum}, Montanaro provides algorithms for both deciding the existence of and outputting solutions to a CSP via a discrete-time quantum walk on the underlying backtracking tree. Unlike in most other quantum walk algorithms \cite{krovi2016quantum,magniez2011search,szegedy2004quantum,shenvi2003quantum}, the input graph is defined implicitly (by the backtracking routine) and therefore not known in advance. Montanaro's existence algorithm is a special case of a quantum walk algorithm by Belovs \cite{belovs2013quantum}, which is able to detect the presence of a marked element in a graph in $\mathcal{O}(\sqrt{W\overline{\eta}})$ steps, where $W$ is the total weight of the graph and $\overline{\eta}$ is its effective resistance. Although this work inspired Montanaro's algorithm, 
the latter achieves a somewhat weaker asymptotic bound in a more restricted set of problems than Belovs' original result suggests might be possible. Hence, we investigate whether quantum backtracking algorithms scale with effective resistance and answer in the affirmative. 

In particular, for a tree of size at most $T$, we describe an algorithm (\cref{alg:find_eta}) that estimates its effective resistance $\overline{\eta}$ and, in doing so, determines whether the tree contains at least one marked vertex. Algorithm 1 has an overall complexity of $\widetilde{\mathcal{O}}(\sqrt{T \overline{\eta}})$, where $\widetilde{\mathcal{O}}$ omits logarithmic factors. As one would expect, and in contrast to Montanaro's result, the decision problem becomes easier as the number of marked vertices is increased. In the worst case, where there are no marked vertices, our algorithm converges to Montanaro's bound of $\widetilde{\mathcal{O}}(\sqrt{T n})$, where $n$ is an upper bound on the depth of the tree.

We then use the effective resistance estimate in \cref{alg:find_marked} to find and output a marked vertex in $\widetilde{\mathcal{O}}(\sqrt{T\overline\eta_{\mathrm{max}}})$ steps. Here, $\overline\eta_{\mathrm{max}}$ is the maximum effective resistance over all subtrees that contain marked vertices. In every case with multiple marked vertices, this achieves a speedup over Montanaro's algorithms. If there are polynomially many marked vertices, this improves Montanaro's algorithm by at least a factor of $n$, where $n$ is a bound on the depth of the tree. Furthermore, this allows us to directly compare the commute time of a classical random walk on the tree to the success rate of the quantum backtracking algorithm.

Our approach shows that the effective resistance of a search tree can be efficiently computed, even without \textit{a priori} knowledge of the tree's structure. Having determined the effective resistance, we are able to perform a walk that finds a marked vertex with an expected $\mathcal{O}({\log(k\overline\eta)})$ measurements. Previously, algorithms for estimating effective resistance have been tailored to situations where at least the set of vertices or edges is known beforehand \cite{wang2017}. Although the method of Ref.\ \cite{wang2017} does not seem immediately applicable to the present problem and we altogether avoid spectral theory, this paper can be loosely viewed as a step towards demonstrating the utility of effective resistance estimates in search problems. Our estimation procedure is similar to those found in  Refs.\ \cite{ito2015approximate} and \cite{jeffery2017quantum}, which also employ amplitude estimation to estimate effective resistance; however, our application of amplitude estimation is arrived at somewhat differently and is not intended to estimate effective resistance to within low multiplicative error. (Achieving low multiplicative error is indeed possible with the algorithms we present, but is not the focus of this paper.)

The algorithm of Ambainis and Kokainis in Ref.\ \cite{ambainis2017quantum} also improves upon Montanaro's complexity bound in certain settings, but by estimating the tree size rather than the effective resistance. The tree size estimation strategy sometimes confers an advantage over estimating effective resistance, particularly if the known upper bound on the tree size is too large. Additionally, some of the analysis of our algorithm is quite similar to that for tree size estimation; however, we apply amplitude estimation to the output of phase estimation whereas Ambainis and Kokainis estimate the phase of the eigenvector closest to the eigenvalue-1 eigenvector. This suggests that an algorithm that integrates the approach of Ref.\ \cite{ambainis2017quantum} with ours may achieve better scaling than either taken independently. Indeed, the results of Ref.\ \cite{jeffery2017quantum} suggest that a tighter bound should be possible, even in the case where the structure of the graph is unknown. 

\subsection{Previous work}

In his approach to the decision problem, Montanaro specialises Belovs' algorithm for determining the existence of marked vertices in a graph \cite{belovs2013quantum,belovs2013time} to rooted trees \cite{montanaro2015quantum}. The algorithm applies phase estimation to the product of two diffusion operators (defined in subsection \ref{sec: diffusion_operators}), using the root of the search tree as the input state. When taken to high enough precision, if one or more vertices in the tree are marked, phase estimation returns the eigenvalue $1$ with probability at least $1/2$. If none of the vertices are marked, the eigenvalue $1$ is seen with probability less than $1/4$. A straightforward application of a Chernoff bound then decides existence. 
\begin{thm}\label{thm:Mont_0}
    Let $\mathcal{T}$ be a tree with root $r$ and vertex set $V(\mathcal{T})$. Let $T$ denote an upper bound on the number of vertices $|V(\mathcal{T})|$, and $n$ an upper bound on the depth of $\mathcal{T}$. Define the functions $f: V(\mathcal{T})\rightarrow \{0,1\}$, where for a vertex $v \in \mathcal{T}$, $f(v) = 1$ if $v$ is marked and $f(v) = 0$ otherwise, and the oracle $h$ such that $h(v)$ returns the set of children of $v$. Then, for any $0 < \delta <1$, there is a quantum algorithm which takes $r$ as input and determines whether $\{v \in V(\mathcal{T}) | f(v) = 1\} = \emptyset$ using $\mathcal{O}({\sqrt{Tn}\log(1/\delta)})$ queries to $f$ and $h$. The algorithm requires $O(1)$ auxiliary operations per query and $\poly(n)$ space, and fails with probability at most $\delta$.
\end{thm}

In the case where one and only one marked vertex is promised to exist, Montanaro measures the output state of the phase estimation procedure described above in the computational basis, then uses the measurement outcome as the new input root. Within an expected $\mathcal{O}({\log n})$ repetitions, this routine returns a marked vertex. Thus, Montanaro proves the following theorem.
\begin{thm}\label{thm:Mont_1}
    Let $\mathcal{T}$, $f$, $h$, $T$, and $n$ be defined as in \cref{thm:Mont_0}. Given the promise that there exists exactly one $m\in V(\mathcal{T})$ such that $f(m) = 1$, there is a quantum algorithm which returns $m$ using $\mathcal{O}({\sqrt{Tn}\log^3(n) \log(1/\delta)})$ queries to $f$ and $h$ for any $0<\delta <1$. The algorithm requires $O(1)$ auxiliary operations per query and $\poly(n)$ space, and fails with probability at most $\delta$.
\end{thm}

In many problems, however, we are not given the promise of a unique marked vertex. Unfortunately, as the number of marked vertices is increased, it becomes increasingly unlikely that measuring the output state of phase estimation in the computational basis returns any state other than the input root. Therefore, for trees with potentially more than one marked vertex, Montanaro finds marked vertices by applying his existence algorithm in conjunction with a classical descent. The search algorithm begins by considering each child $c$ of the root vertex $r$ in turn and determining whether at least one marked vertex exists in the subtree $\mathcal{T}(c)$ rooted at $c$ (using the existence routine). For a child $c_0$ of $r$ for which the subtree $\mathcal{T}(c_0)$ is found to contain at least one marked vertex, the algorithm checks whether $c_0$ is marked. If so, $c_0$ is returned, and if not, the procedure is repeated on the children of $c_0$. Once the algorithm finds a marked vertex, it returns the vertex, unmarks it, and restarts at the root, until all $k$ vertices are discovered. Assuming that the degree of every vertex in $\mathcal{T}$ is $O(1)$, the following bounds are thus achieved.

\begin{thm}\label{thm:Mont_2}
     For $\mathcal{T}$, $f$, $h$, $T$, and $n$ defined as in \cref{thm:Mont_0} and any $0 < \delta <1$, there is a quantum algorithm which returns $v$ such that $f(v) =1$ or ``none'' if no such $v$ exists using $\mathcal{O}({\sqrt{T}n^{3/2}\log (n) \log(1/\delta)})$ queries to $f$ and $h$. The algorithm requires $O(1)$ auxiliary operations per query and $\poly(n)$ space, and fails with probability at most $\delta$. It can be repeated to find all $k$ marked vertices using $\mathcal{O}\left(k\sqrt{T}n^{3/2}\log (n)\log(k/\delta)\right)$ queries to $f$ and $h$. 
\end{thm}

Ambainis and Kokainis reduce the number of steps needed for trees with multiple marked vertices by introducing a tree-size estimation algorithm in Ref.\ \cite{ambainis2017quantum}. The result is similar to \Cref{thm:Mont_2} with the upper bound $T$ on the size of the tree replaced by the number of vertices $T'$ actually visited by a classical backtracking algorithm.

\subsection{Main results}

In Ref.\ \cite{montanaro2015quantum}, Montanaro leaves open the question of whether an algorithm similar to that of \Cref{thm:Mont_1} can be used to obtain a similar bound for trees with an arbitrary number of marked vertices. We proceed to such a generalisation in the present work.

We first establish the following theorem, which utilises the effective resistance $\overline\eta$ between the root vertex of the input tree $\mathcal{T}$ and the set of marked vertices in $\mathcal{T}$ (which we will refer to as the ``effective resistance of $\mathcal{T}$''). Note that for trees with many marked vertices, this effective resistance can be substantially smaller than the resistance between $r$ and a particular marked vertex, but never less than $1/d_r$ where $d_r$ is the degree of the root. 
\begin{thm}\label{thm: k_exist}
     For $\mathcal{T}$, $f$, $h$, $T$, and $n$ defined as in \cref{thm:Mont_0} and any $0 < \delta <1$, there is a quantum algorithm which determines whether $\{v \in V(\mathcal{T}) | f(v) = 1\} = \emptyset$ in $\mathcal{O}({\sqrt{T\overline\eta}\log(1/\delta)})$ queries to $f$ and $h$, where $\overline{\eta}$ is the effective resistance of $\mathcal{T}$ (not necessarily known beforehand). The algorithm requires $O(1)$ auxiliary operations per query and $\poly(n)$ space, and fails with probability at most $\delta$. 
\end{thm}

We use \Cref{thm: k_exist} to bound the number of steps taken to find a marked vertex. Leting $\mathcal{T}(v)$ denote the subtree of $\mathcal{T}$ rooted at $v$ and $\overline{\eta}(v)$ the effective resistance between $v$ and the set of marked vertices in $\mathcal{T}(v)$, we define the \emph{maximum} effective resistance $\overline{\eta}_{\mathrm{max}}$ as the largest finite $\overline{\eta}(v)$ over all subtrees $\mathcal{T}(v)$.
The polylogarithmic factors in the theorem below depend on $\overline{\eta}$ ($=\overline\eta(r)$), which may be smaller than $\overline\eta_{\max}$.
\begin{thm}\label{thm: k_marked}
     Let $\mathcal{T}$, $f$, $h$, $T$, $n$, and $\overline\eta$ be defined as in \cref{thm: k_exist}. Let $\eta(v)$ denote the effective resistance of the subtree rooted at $v$, and define \begin{equation} \label{eta_max} \overline{\eta}_{\mathrm{max}} = \sup_{\substack{v\in V(\mathcal{T}) \\ \overline{\eta}(v) < \infty}}\overline{\eta}(v). \end{equation} Then, for any $0 < \delta < 1$, there is a quantum algorithm which returns a $v \in V(\mathcal{T})$ such that $f(v) = 1$ or ``none'' if no such $v$ exists using $\mathcal{O}(\sqrt{T\overline\eta_{\mathrm{max}}}\log^4 (k\overline\eta) \log(1/\delta))$ queries to $f$ and $h$. The algorithm requires $O(1)$ auxiliary operations per query and $\poly(n)$ space, and fails with probability at most $\delta$. 
\end{thm}
If desired, the algorithm of the above theorem can be repeated $k$ times to return all $k$ marked vertices, resulting in an overall runtime of $\mathcal{O}(k\sqrt{T n}\log^4 (k n)\allowbreak \log(k/\delta))$. Although this is a loose bound, based on the fact (proven in Section \ref{sec:math}) that the effective resistance is upper bounded by the depth of the tree, it already achieves an improvement over \Cref{thm:Mont_2}. To tighten the bound, one would need to take into account the change in effective resistance that results from unmarking marked vertices that have been found. This is an interesting question in its own right, but beyond the scope of this paper. In \Cref{sec:grover}, we explore an example problem (Grover search) in which the bound can be tightened and show that in at least some cases our algorithm achieves optimal scaling.

\section{Preliminaries}

\subsection{Setting}
We consider a tree $\mathcal{T}$ of bounded degree, with vertex set $V(\mathcal{T})$ and edge set $E(\mathcal{T})$. We assume access to the following:
\begin{enumerate}
    \item the root vertex $r$ of $\mathcal{T}$,
    \item upper bounds on the depth and size of $\mathcal{T}$,
    \item an upper bound $d$ on the degree of each vertex in $\mathcal{T}$,
    \item an oracle $h$ which, for each $v \in V(\mathcal{T})$, returns a tuple that encodes the children of $v$,
    \item an oracle which efficiently evaluates a function $f:V(\mathcal{T}) \longrightarrow \{0,1\}$.
\end{enumerate}

We call a vertex $v \in V(\mathcal{T})$ ``marked'' if $f(v) = 1$ and ``not marked'' if $f(v)=0$.  For simplicity, we will assume that $f(r) = 0$, i.e., that the root vertex is never marked. This could always be ensured by first checking whether the root is marked and if so, returning and unmarking it.

The purpose of our algorithm is to find the marked vertices in $\mathcal{T}$. Each run of the algorithm returns one of the ``shallowest" marked vertices, i.e., a marked vertex that has no marked ancestors. In a standard CSP, such as the Boolean satisfiability problem (SAT), where marked vertices can be associated with satisfying solutions, it is typically the case that if a particular vertex is marked, then so are all of its descendants. In such contexts, therefore, only the set of these shallowest marked vertices is of interest. In the abstract, however, one can easily imagine a setting where, despite a vertex being marked, a subset of its descendants are not. Our method is able to find \textit{all} of the marked vertices in either scenario, simply by unmarking a marked vertex once it is found, thereby giving access to its next shallowest marked descendants (if any) in the next run.

\subsection{Notation}

We use the following notation to describe the input tree $\mathcal{T}$: 

\begin{itemize}
    \item $T$ denotes an upper bound on the number of vertices in $\mathcal{T}$. 
    \item $n$ denotes an upper bound on the depth of $\mathcal{T}$.
    \item $d$ denotes an upper bound on the degree of any vertex in $\mathcal{T}$. 
    \item $\ell_{v}$ denotes the depth of a vertex $v \in V(\mathcal{T})$, and $d_{v}$ denotes its degree (accordingly, $\ell_{v} \leq n$ and $d_{v} \leq d$ for all $v \in V(\mathcal{T})$).
    \item $k$ is the total number of marked vertices in $\mathcal{T}$.
    \item For any vertex $v \in V(\mathcal{T})$, we write $\tree{v}$ to denote the subtree of $\mathcal{T}$ rooted at $v$.
    \item For two vertices $v$ and $u$ where $u$ is a vertex in $\tree{v}$, we write $\mathcal{P}(v,u)$ to denote the (shortest) path in $\mathcal{T}$ connecting $v$ to $u$.
    \item $c \leftarrow v$ indicates that vertex $c$ is a child of vertex $v$. A summation indexed by ``$c \leftarrow v$" is hence a sum over all of the children of vertex $v$. 
\end{itemize}

Additionally, we adopt a few conventions for ease of presentation. First, where we expect no confusion, we abusively write $u \in \tree{v}$ instead of $u \in V(\tree{v})$ and $(u,w) \in \tree{v}$ instead of $(u,w) \in E(\tree{v})$. That is, we use a single identifying variable to represent a vertex and a pair to represent an edge, dropping the explicit specification of the appropriate set. We also follow the convention that an empty sum evaluates to zero. Finally, we denote the set of all ``shallowest" marked vertices in the input tree $\mathcal{T}$ by $\mathcal{M}$. This is the subset of marked vertices that do not have any marked ancestors. To be precise,
\[ \mathcal{M} \equiv \{m \in V(\mathcal{T}): f(m) = 1, f(v) = 0 \enspace \forall v \in \mathcal{P}(r,m)\setminus \{m\}\}. \] 
We then define $\mathcal{M}(v) \equiv \mathcal{M} \cap V(\tree{v})$ for any $v \in \mathcal{T}$ as the set of shallowest marked vertices in the subtree $\tree{v}$ rooted at $v$.

\subsection{Phase and Amplitude estimation}

Our algorithm, like that in Ref.\ \cite{montanaro2015quantum}, applies quantum phase estimation on a particular unitary operator, with input state $\ket{r}$. We recall the statement of phase estimation from Ref.\ \cite{montanaro2015quantum}.

\begin{thm}[phase estimation \cite{montanaro2015quantum}] \label{thm: phase_estimation}
    For every integer $s \geq 1$ and every unitary operator $U$ on $m$ qubits, there exists a uniformly generated quantum circuit $C$ such that $C$ acts on $m + s$ qubits and the following hold.
    \begin{enumerate}
        \item $C$ uses the controlled-$U$ operator $\mathcal{O}({2^{s}})$ times and contains $\mathcal{O}({s^{2}})$ other gates.
        \item If $U\ket{\psi} = \ket{\psi}$, then $C\ket{\psi}\ket{0^{s}} = \ket{\psi}\ket{0^{s}}$.
        \item If $U\ket{\psi} = e^{2i\theta}\ket{\psi}$ with $\theta \in (0,\pi)$, then $C \ket{\psi}\ket{0^{s}} = \ket{\psi}\ket{\omega}$, where 
        \[|\bra{\omega}0^{s}\rangle|^{2} = \sin^{2}(2^{s}\theta)/(2^{2s}\sin^{2}\theta).\]
        \item For any $\ket{\phi} = \sum_{j}\lambda_{j}\ket{\psi_{j}}  \in (\mathbb{C}^{2})^{\otimes m}$, where $U\ket{\psi_{j}} = e^{2i\theta_{j}}\ket{\psi_{j}}$,
        \[ C\ket{\phi}\ket{0^{s}} = \sum_{j}\lambda_{j}\ket{\psi_{j}}\ket{\omega_{j}} \]
        with $\sum_{j:\,\theta_{j}\geq \epsilon}|\bra{\omega_{j}}0^{s}\rangle|^{2} = \mathcal{O}({1/(2^{s}\epsilon)})$.
    \end{enumerate}
\end{thm}

Amplitude estimation, which is also used in our procedure, is based on phase estimation \cite{brassard2002quantum}. We state it a bit differently below since we will not need to analyse the resulting quantum state. Rather, we are mostly interested in determining whether the squares of probability amplitudes are within a constant of ${1}/{2}$, which is quite similar to distinguishing the state $\ket{0^s}$ from the rest of the state space of the ancilla qubits.
\begin{thm}[amplitude estimation] \label{thm:amplitude_estimation}
    For every integer $s \geq 1$ and every unitary operator $U$ on $m$ qubits, there exists a uniformly generated quantum circuit $C$ such that $C$ acts on $m + s$ qubits and the following hold.
    \begin{enumerate}
        \item $C$ uses the controlled-$U$ operator $\mathcal{O}({2^s})$ times and contains $\mathcal{O}({s^2})$ other gates.
        \item If $U\ket{\psi_\mathrm{bad}} = \ket{\psi_\mathrm{bad}}$, then $C\ket{\psi_\mathrm{bad}}\ket{0^s} = \ket{\psi_\mathrm{bad}}\ket{0^s}$. Similarly, if $U\ket{\psi_\mathrm{good}} = -\ket{\psi_\mathrm{good}}$, then $C\ket{\psi_\mathrm{good}} \ket{0^s} = \ket{\psi_\mathrm{good}} \ket{\pi/2}$.
        \item If $\ket\psi = \sin\theta \ket{\psi_\mathrm{good}} + \cos\theta\ket{\psi_\mathrm{bad}}$ for $\theta \in (0,\pi)$, then a measurement of the second register of $C\ket{\psi}\ket{0^s} = \sum_i \ket{\psi_i} \ket{\theta_i}$ returns a state $\ket{\widetilde\theta}$ such that for $\epsilon = {k}/{2^s}$ and some integer $1< k \leq 2^s$.
            \begin{enumerate}
                \item $\Pr{\abs*{\widetilde\theta - \theta} \leq \epsilon} = 1-\mathcal{O}({{1}/({2^{s}\epsilon})})$ and
                \item $\Pr{\abs{\theta - \overline\theta} > 2\epsilon \vert \abs{\widetilde\theta - \overline\theta} \leq \epsilon} = \mathcal{O}({{1}/({2^s\epsilon})})$ for some constant $\overline\theta \in [0,\pi)$, and where the differences are taken \!\!\!$\mod \pi$.
            \end{enumerate}
    \end{enumerate}
\end{thm}

\subsection{Diffusion operators} \label{sec: diffusion_operators}
Our algorithm finds marked vertices in $\mathcal{T}$ using a quantum walk on the Hilbert space spanned by $V(\mathcal{T})$, with initial state $\ket{r}$. This quantum walk is effected by a set of \textit{diffusion operators}, adapted from Refs.\ \cite{ambainis2017quantum}, \cite{belovs2013quantum}, and \cite{montanaro2015quantum}.

Let $\eta = {2^i}/{d}$ be an input parameter to the algorithm, where $i \in [\left\lceil\log (d n)\right\rceil]$. For each vertex $v \in \mathcal{T}$, the corresponding diffusion operator $D_{v}$ is defined as follows:
\begin{itemize}
    \item If $v \in \mathcal{M}$, then $D_{v}$ is the identity operator.
    \item If $v \not\in \mathcal{M}$, then
    \[ D_{v}(\eta) \equiv I - 2\ket{\psi_{v}(\eta)}\bra{\psi_{v}(\eta)},\]
    where for $v \in \mathcal{T}$,
    \begin{equation} \label{def: psi_v}  \ket{\psi_{v}(\eta)} \equiv  \begin{dcases} \frac{1}{\sqrt{1 + d_{v}\eta}}\left(\ket{v} + \sqrt{\eta}\sum_{c\leftarrow v}\ket{c}\right) \qquad &\text{if $v = r$} \\
    \frac{1}{\sqrt{d_{v}}}\left(\ket{v} + \sum_{c\leftarrow v}\ket{c}\right) \qquad  &\text{if $v\neq r$}.
    \end{dcases} \end{equation}
\end{itemize}

Each diffusion operator $D_{v}$ acts only on the subspace spanned by the vertex $v$ and its children, and can therefore be implemented with only local knowledge. Specifically, it requires access to only the oracles $f$ and $h$: $f$ is evaluated at $v$ to check whether $v$ is marked, and $h$ gives the children of $v$.

Our diffusion operators are identical to those in Ref.\ \cite{montanaro2015quantum} with the exception of $D_{r}$, which is defined with respect to the parameter $\eta$ instead of the maximum depth $n$. This parameter will be tuned in \Cref{alg:find_eta} and, as will be shown in \Cref{sec:math}, its ``optimal'' value is equal to the effective resistance of $\mathcal{T}$.

Now, let $A$ be the set of vertices that are at even distances from the root (including the root itself), and let $B$ be the set of vertices at odd distances from the root. We define the unitary operators
\[ R_{A}(\eta) \equiv \bigoplus_{v\in A}D_{v}(\eta) \qquad \text{and} \qquad R_{B} \equiv \ket{r}\bra{r} + \bigoplus_{v\in B}D_{v}(\cdot) \]
as direct sums of diffusion operators. Because $R_A(\eta)$ is just a rotation of $R_A(n)$ in the subspace $\mathrm{span}(\left\{\ket{r}, \sum_{c\leftarrow r} \ket{c}/{\sqrt{d_r}} \right\})$ for $\mathcal{O}({\log(n)})$ different choices of $\eta$, these operators can be implemented efficiently given that $R_A(n)$ can be implemented efficiently (as is explicitly demonstrated in Ref.\ \cite{montanaro2015quantum}). An application of the operator $R_{B}R_{A}(\eta)$ amounts to a step in the quantum walk.

\section{Algorithms}\label{sec:algorithm}

In Ref.\ \cite{montanaro2015quantum}, Montanaro proposes an algorithm for trees in which a unique marked vertex is promised to exist. Given such a tree and its root vertex as the input state, phase estimation on the operator $R_{B}R_{A}(\eta)$, defined as in subsection \ref{sec: diffusion_operators} except with $\eta$ set to $n$, is used to approximately produce a superposition state over the vertices on the path from the root to the unique marked vertex. With probability $1/2$, measuring this superposition in the computational basis collapses it to some state corresponding to a non-root vertex. This procedure is then repeated with the resulting state as the input until a marked vertex is reached. Conditioned on measuring away from the root, any non-root vertex along the path is sampled with equal probability. As a result, the number of steps to reach a marked vertex is logarithmic in the depth $n$ (in contrast to the polynomial overhead that would be incurred by classical descent).

Our algorithm is essentially a generalisation of this scheme to trees containing an arbitrary number of marked vertices (not necessarily known beforehand), and uses the same basic framework. \Cref{alg:find_marked} creates a superposition state $\ket{\Phi}$ with support only on vertices that either are marked or have marked descendants. Unlike in the restricted case of a single marked vertex, na\"ively running Montanaro's algorithm with $\eta = n$ on a tree with multiple marked vertices may, depending on the distribution of the marked vertices in the tree, result in a state for which the probability $\abs{\braket{\Phi}{r}}^2$ of measuring $\ket{r}$ is too large. It turns out that with $\eta$ set to the effective resistance of the tree, $\ket{\Phi}$ can be prepared with fixed precision and $\abs{\braket{\Phi}{r}}^{2}$ can be kept close to $1/2$.

However, since the structure of the input tree and the distribution of the marked vertices are not known, the effective resistance $\overline{\eta}$ cannot be determined in advance. \Cref{alg:find_eta} describes the method by which we obtain an estimate $\widetilde\eta$ such that $\abs{\widetilde\eta - \overline\eta} \leq \mathcal{O}({\Delta \overline\eta})$ for a fixed $\Delta$. For a given upper bound $d$ on the degree of the tree, we begin by assuming $\eta \approx {1}/{d}$ and then increment $\eta$ until it hits $\widetilde\eta$ or its maximum possible value of $n$. Provided that the algorithm does not exit early, which we will show is exceptionally unlikely, the interpolation crosses an $\eta$ ``close enough'' to $\overline\eta$ with certainty if such an $\eta$ exists. Since the last iteration of the loop effectively implements Montanaro's existence algorithm from \cite{montanaro2015quantum}, if $\Delta$ is taken to be sufficiently small, the last iteration will produce $\infty$ if there are no marked vertices. If we are interested in a more precise estimate of $\overline\eta$, after \Cref{alg:find_eta} returns some estimate $\widetilde\eta$ such that $\widetilde\eta/\overline\eta \approx 1$, we can perform amplitude estimation to higher precision but with initial guess $\widetilde\eta$. Alternatively, we can replace $2^i$ in \Cref{alg:find_eta} with $S^i$, where $S \in (1,2)$. Because it is sufficient for our application to backtracking algorithms, we assume that a multiplicative step size of 2 produces the desired level of precision.

\begin{algorithm}[H]
\caption{Estimate effective resistance, \textbf{\textit{Estimate-Res}}($v$)}
\begin{algorithmic}[1]\label{alg:find_eta}
    \Require A vertex $v\in\mathcal{T}$; upper bounds $n$, $T$, and $d$ on the depth, number of vertices, and degree of $\mathcal{T}$, respectively; unitaries $R_{A}(\cdot)$ and $R_{B}$; a failure probability $\delta_0$; universal constants $\gamma_1,\gamma_2$.
    \State $i \gets 0$. 
    \State $\eta \gets \min\left\{{2^i}/{d},n\right\}$.
    \State Let $C(\eta)\ket{v}\ket{0^s} = \sin(\beta_\eta)\ket{v}\ket{0^s} + \cos(\beta_\eta) \ket{\psi_{\mathrm{bad}}}$ be the output of phase estimation with unitary $R_BR_A(\eta)$, input state $\ket{v}$, and $s$ ancilla qubits in the second register. 
    \State Perform amplitude estimation on the second register $\gamma_1\log(1/\delta_0)$ times with $\ket{\psi_{\text{good}}} = \ket{0^s}$ and precision ${\gamma_2}$. Let the output be the multiset $\widetilde\Theta = \{\widetilde{\beta}_{j}\}_{j}$.
    \State If  more than half of the elements $\widetilde\beta_j \in \widetilde\Theta$ satisfy $\abs{\widetilde\beta_{j} - \pi/4} \leq \pi/16$,  
    \Return $\widetilde\eta = \eta \cot^2 \widetilde\beta$ where $\widetilde\beta$ is the most frequent element in $\widetilde\Theta$.
    \State If $\eta = n$, \Return $\infty$.
    \State Increment $i$, go to Step 2.
\end{algorithmic}
\end{algorithm}

\Cref{alg:find_eta} is implemented in a loop in \Cref{alg:find_marked}, which uses effective resistance estimates to find a marked vertex as follows. Starting at the root vertex $r$, we obtain an estimate $\widetilde\eta$ using \textbf{\textit{Estimate-Res}}($r$). Given that $\widetilde\eta \sim\overline{\eta}$, the output $\ket{\Phi}$ of phase estimation on unitary $R_{B}R_{A}(\overline{\eta})$ with input state $\ket{r}$ is expected to satisfy $\abs{\braket{r}{\Phi}^2 - {1}/{2}} \leq \Delta$ with high probability. By measuring $\ket\Phi$ in the computational basis and setting the output as the new root, we progress down the tree towards a marked vertex. \Cref{sec:expected_steps} demonstrates that the expected number of steps in this walk is $\mathcal{O}({\log(k \overline\eta)})$.

\begin{algorithm}[H]
\caption{Find a marked vertex}
\begin{algorithmic}[1]\label{alg:find_marked}
    \Require The root vertex $r$ of $\mathcal{T}$; upper bound on the number of vertices in $\mathcal{T}$; unitaries $R_{A}(\cdot)$ and $R_{B}$; other inputs required for \textbf{\textit{Estimate-Res}}($r$).
    \State $v \gets r$.
    \State If $f(v)=1$, \Return $v$.
    \State $\widetilde\eta \gets \textbf{\textit{Estimate-Res}}(v)$.
    \While{$\widetilde\eta \neq \infty$}
        \State Let $\ket{\Psi} = C(\widetilde\eta)\ket{v}\ket{0^s} = \sum_i \alpha_i\ket{\psi_i}\ket{\omega_i}$ be the output of phase estimation with unitary {\color{white} bla bl \hspace{.015em}} $R_B R_A(\widetilde{\eta})$, input state $\ket{r}$, and precision $\mathcal{O}({\sqrt{1/T\widetilde{\eta}}})$. 
        \State Measure the second register of $\ket\Psi$ in the computational basis. Let the output be $\ket\omega$.
        \State If $\omega \neq 0^s$, \textbf{continue}. 
        \State Measure the first register of $\ket\Psi$. Let the output be $\ket{o}$.
        \State $v \gets o$.
        \State If $f(v) = 1$, \Return $v$.
        \State $\widetilde\eta \gets \textbf{\textit{Estimate-Res}}(v)$.
    \EndWhile
    \State \Return ``no marked vertex.''
\end{algorithmic}
\end{algorithm}

\section{Analysis}
In this section, we analyse \Cref{alg:find_eta,alg:find_marked}. Subsection \ref{sec:state_features} defines a particular eigenstate $\ket{\Phi}$ of the unitary $R_{B}R_{A}(\eta)$ that will be used to find marked vertices, and subsection \ref{sec: prepare_phi} describes how phase estimation can be applied to prepare such an eigenstate. In subsection \ref{sec:alg_analysis}, we evaluate the accuracy and efficiency of \Cref{alg:find_eta} in returning a useful value for the parameter $\eta$. We bound the expected runtime of \Cref{alg:find_marked} in subsection \ref{sec:alg_analysis2}.

As will be made more evident in subsection \ref{sec:alg_analysis}, in the case of no marked vertices, our algorithm reduces to an algorithm similar to that of \Cref{thm:Mont_0}. Therefore, we assume in the following discussion that there exists at least one marked vertex in the search tree $\mathcal{T}$.

\subsection{On the path to marked vertices}\label{sec:state_features}

For each marked vertex $m \in \mathcal{M}$, the state
\[ \ket{\phi_{m}} = \sqrt{\eta}\ket{r} + \sum_{\substack{v \in \mathcal{P}(r,m)\\v \neq r}}(-1)^{\ell_{v}}\ket{v} \]
is an eigenvector of $R_{B}R_{A}(\eta)$ with eigenvalue 1. This can be seen from the fact that for a given $\eta$, $\ket{\phi_{m}}$ as defined above is orthogonal to $\ket{\psi_{v}(\eta)}$ for all $v \in \mathcal{T}$ [cf.\ Eq.\ \eqref{def: psi_v}]. $\ket{\phi_{m}}$ encodes the entire path $\mathcal{P}(r,m)$ from the root $r$ to the marked vertex $m$.

We introduce the state 
\begin{equation}\label{def:Phi1}
\ket{\Phi} \equiv \sum_{m \in \mathcal{M}}C_{m}\ket{\phi_{m}} =  \sqrt{\eta}\sum_{m\in\mathcal{M}}C_{m}\ket{r} + \sum_{m\in\mathcal{M}}C_{m}\sum_{\substack{v \in \mathcal{P}(r,m)\\v\neq r}}(-1)^{\ell_{v}}\ket{v}
\end{equation}
as a normalised superposition of these ``path" eigenvectors with real coefficients $(C_{m})_{m\in\mathcal{M}}$. Clearly, $\ket{\Phi}$ has support only on the shallowest marked vertices in $\mathcal{T}$ and their ancestors. Accordingly, we define the ``solution tree" $\widetilde{\mathcal{T}}$ to be the largest subtree of $\mathcal{T}$ in which all leaves are the shallowest marked vertices on their respective paths, i.e.,
\[ \widetilde{\mathcal{T}} \equiv \bigcup_{m\in\mathcal{M}}\mathcal{P}(r,m). \] 
We use $\widetilde{\mathcal{T}}(v)$ to indicate the subtree of $\widetilde{\mathcal{T}}$ rooted at a vertex $v \in \widetilde{\mathcal{T}}$.

Letting
\begin{equation} \label{def: sin_beta} \sin\beta \equiv \bra{r}\Phi\rangle = \sqrt{\eta}\sum_{m\in\mathcal{M}}C_{m}, \end{equation}
and defining $\kappa:V(\mathcal{T})\longrightarrow \mathbb{R}$ by
\begin{equation}\label{def:kappas}  \kappa_v \equiv \kappa(v)  = \frac{1}{\cos\beta}\sum_{m\in\mathcal{M}(v)} C_{m}, 
\end{equation}
we rewrite Eq.\ \eqref{def:Phi1} as 
\begin{equation} \label{def:Phi2}
\ket{\Phi} = \sin\beta\ket{r} + \cos\beta\sum_{\substack{v\in\mathcal{T}\\v\neq r}}(-1)^{\ell_{v}}\kappa_{v}\ket{v}.
\end{equation}

By construction, $\kappa$ satisfies
\begin{equation} \label{E: kappas_recurrence}
    \kappa_{v} = \begin{dcases} \kappa_{v} \qquad &\text{if $v\in\mathcal{M}$} \\
    \sum_{c\leftarrow v}\kappa_{v} \qquad &\text{if $v \not\in\mathcal{M}$} \end{dcases}
\end{equation}
with $\kappa_{v} = 0$ for $v \not\in\widetilde{\mathcal{T}}$ (using the convention that an empty sum is equal to zero), as well as normalisation:
\begin{equation} \label{E: kappas_norm}
    \sum_{\substack{v\in\mathcal{T}\\v\neq r}}\kappa_{v}^{2} = 1.
\end{equation}

We further constrain $\kappa$ as follows: for every vertex $v \in \widetilde{\mathcal{T}}$ and any marked vertices $m_{0}, m_{1} \in \mathcal{M}(v)$,
\begin{equation} \label{E: path_leaf} \sum_{u\in\mathcal{P}(v, m_{0})}\kappa_{u} = \sum_{u\in\mathcal{P}(v,m_{1})}\kappa_{u}. \end{equation}
In other words, given any vertex $v$ in the solution tree $\widetilde{\mathcal{T}}$, the weight in $\ket{\Phi}$ on the path from $v$ to a marked vertex in $\widetilde{\mathcal{T}}(v)$ is independent of the choice of marked vertex. One can easily see that as a consequence of this condition, $\kappa_{m} \neq 0$ for at least some $m \in \mathcal{M}$. If this were not the case, Eq.\ \eqref{E: kappas_recurrence} would imply $\kappa_{v} = 0$ for all $v \in \mathcal{T}$, violating Eq.\ \eqref{E: kappas_norm}.

As will become clear, a set of coefficients satisfying Eqs. \eqref{E: kappas_recurrence}-\eqref{E: path_leaf} is guaranteed to exist and are unique (up to a global phase).\footnote{The reader familiar with the theory of electrical networks and flows may already see how Eq.\ \eqref{E: path_leaf} can naturally follow from considering a voltage between the root and the set of marked leaves. We nonetheless avoid this theory, as it obfuscates some of the ideas used in our walk.} (For consistency with Eq.\ \eqref{def:kappas}, we choose them to all be real.) In particular, the constraint imposed by Eq.\ \eqref{E: path_leaf}, which will justified in the following analyses, allows us determine some interesting properties of the solution tree $\widetilde{\mathcal{T}}$. We state some of the results relevant to this section below, postponing their proofs until \Cref{sec:math}.

\cref{cor: kappa_recursion} of \Cref{sec:math} demonstrates that $\kappa$ satisfies Eqs.\ \eqref{E: kappas_recurrence}-\eqref{E: path_leaf} if and only if for every $m_{0} \in \mathcal{M}$,
\begin{equation} \label{E: kappa_m} \left(\sum_{m\in\mathcal{M}}\kappa_{m}\right)  \sum_{\substack{v\in\mathcal{P}(r,m_{0})\\v\neq r}}\sum_{m'\in\mathcal{M}(v)}\kappa_{m'} = 1. \end{equation}
This constitutes a system of $|\mathcal{M}|$ equations that determines $(\kappa_{m})_{m\in\mathcal{M}}$ and therefore $\kappa$ uniquely [cf.\ Eq.\ \eqref{E: kappas_recurrence}]. An important implication of Eq.\ \ref{E: kappa_m} is that $\kappa$ is independent of the input parameter $\eta$ and depends only on the structure of $\widetilde{\mathcal{T}}$. In particular, we will show in \Cref{thm: eta_bound} of \Cref{sec:math} that $\kappa_{r}^{2}$ is an invariant of finite trees and that its value is bounded in terms of the number of marked vertices, the degree of the root vertex, and the (maximum) depth of the tree as
\begin{equation} \label{E: eta_bound} \frac{1}{\kappa_{r}^{2}} \in \left[\max\left(\frac{1}{k}, \frac{1}{d_r} \right), n\right]. \end{equation}
As will be further discussed in subsection \ref{sec:alg_analysis}, we exploit this result in designing \Cref{alg:find_eta}, which tunes the input parameter $\eta$ so as to optimise the superposition state $\ket{\Phi}$ for our intended purpose. That is, if we are able to produce a state $\ket\Phi$ such that $\braket{r}{\Phi}$ is not too large, a state closer to the marked state can be revealed by measuring $\ket\Phi$ in the computational basis. Indeed, it follows from Eqs.\ \eqref{def: sin_beta} and \eqref{def:kappas} that 
\begin{equation}\label{eqn:angle}
    \tan\beta = \sqrt{\eta}\kappa_r.
\end{equation}
This implies that if there exists at least one marked vertex, then there is a particular value of $\eta$, determined by the structure of the search tree, for which the the probability $\abs{\braket{r}{\Phi}}^2$ of sampling $\ket{r}$ is approximately 1/2. \Cref{sec:alg_analysis} shows that when $\eta$ is equal to the effective resistance $\overline{\eta}$ of the tree, the probability of sampling the root vertex is 1/2.\footnote{In \Cref{sec:alg_analysis}, the relative amplitude in the root and non-root states is characterized by $\cot^2\beta$. We call $\eta$ achieving $\cot^2\beta =1$ ``optimal'' in that it assigns equal probability to sampling either a root or non-root state.} In the following subsection, we demonstrate that the particular superposition of eigenvectors $\ket{\Phi}$ defined here can be approximately produced via phase estimation.

\subsection{Complexity of preparing $\ket{\Phi}$}\label{sec: prepare_phi}

This subsection primarily follows Ref.\ \cite{montanaro2015quantum}, with relevant adjustments as they become necessary in the generalisation to trees with possibly more than one marked vertex.

For a fixed value of $\eta$ and given $\ket{\Phi}$ as defined in the preceding subsection, let $\ket{\Phi^{\perp}}$ denote the particular normalised state orthogonal to $\ket{\Phi}$ such that 
$\ket{r} = \bra{r}\Phi\rangle\ket{\Phi} + \bra{r}\Phi^{\perp}\rangle\ket{\Phi^{\perp}}$, that is,
\begin{equation} \label{def: Phi_perp} \ket{\Phi^{\perp}} \equiv \cos\beta\ket{r} - \sin\beta\sum_{\substack{v\in\mathcal{T}\\v\neq r}}(-1)^{\ell_{v}}\kappa_{v}\ket{v}. \end{equation}
By direct calculation,
\begin{equation} \label{E: phi_perp_orth} \bra{\Phi^{\perp}}\phi_{m}\rangle = 0 \end{equation}
for all $m \in \mathcal{M}$.

We perform phase estimation on unitary $R_{B}R_{A}(\eta)$ with $s$ ancilla qubits and input state 
\[ \ket{r} = \sin\beta\ket{\Phi} + \cos\beta\ket{\Phi^{\perp}}, \] 
which yields a state of the form
\begin{equation}\label{eqn:Psi} 
\ket\Psi = \sin\beta\ket{\Phi}\ket{0^{s}} + \cos\beta\sum_{j}\lambda_{j}\ket{\psi_{j}}\ket{\omega_{j}}. 
\end{equation}
By virtue of Eq.\ \eqref{E: phi_perp_orth}, we can take the eigenvectors $\ket{\psi_{j}}$ in the expansion of $\ket{\Phi^{\perp}}$ to be orthogonal to the space spanned by the path eigenvectors, i.e., such that $\bra{\Phi}\psi_{j}\rangle = 0$ for all $j$. Decomposing each $\ket{\omega_{j}}$ as $\ket{\omega_{j}} = \mu_{j}\ket{0^{s}} + \ket{\omega_{j}'}$, where the $\ket{\omega_{j}'}$ are subnormalised vectors orthogonal to $\ket{0^{s}}$, the probability of obtaining $\ket{0^{s}}$ upon measuring the second register is
$\sin^{2}\beta + \cos^{2}\beta\sum_{j}|\lambda_{j}\mu_{j}|^{2}$,
in which case the first register collapses to
\begin{equation} \label{Phi_tilde} \ket{\widetilde{\Phi}} = \frac{1}{\sqrt{\sin^{2}\beta + \cos^{2}\beta\sum\limits_{j}|\lambda_{j}\mu_{j}|^{2}}}\left(\sin\beta\ket{\Phi} + \cos\beta\sum_{j}\lambda_{j}\mu_{j}\ket{\psi_{j}}\right). \end{equation}

It follows immediately from \Cref{thm: phase_estimation} that 
\begin{equation}\label{eqn: lower_bound}
 \sum_{j:\,\theta_{j}\geq \epsilon}|\lambda_{j}\mu_{j}|^{2} \leq \sum_{j:\,\theta_{j}\geq \epsilon}|\mu_{j}|^{2} = \bigO{\frac{1}{2^{s}\epsilon}}.
\end{equation}

In order to distinguish the state $\ket{0^s}$ from all other states, we also need to bound the terms indexed by $j$ such that $0 < \theta_j < \epsilon$. We prove the following lemma in \Cref{xi_proof}.
\begin{restatable}{lemma}{Pepsilon} \label{lem: P_epsilon} Let $P_{\epsilon}$ denote the projector onto $\mathrm{span}(\{\ket{\psi}:R_{B}R_{A}(\eta)\ket{\psi} = e^{2i\theta}\ket{\psi}, |\theta| \leq \epsilon\})$. Then,
\[\|P_{\epsilon}\ket{\Phi^{\perp}}\| = \bigO{\epsilon\sqrt{T\eta}} \]
\end{restatable}

This lemma allows us to obtain
\begin{equation}\label{eqn: small_theta}
    \sum_{j:\,\theta_{j} \leq \epsilon}|\lambda_{j}\mu_{j}|^{2} \leq \sum_{j:\,\theta_{j} \leq \epsilon}|\lambda_{j}|^{2} = \|P_{\epsilon}\ket{\Phi^{\perp}}\|^{2} = \bigO{\epsilon^{2}T\eta}.
\end{equation} 

Thus, for a fixed precision $\delta$, taking $\epsilon = \Theta(\delta/\sqrt{T \eta})$  and $2^s = \Theta(\sqrt{T\eta}/\delta^3)$ in the bounds of Eqs. \eqref{eqn: lower_bound} and \eqref{eqn: small_theta} yields
\begin{equation} \label{delta_squared} \sum_{j}|\lambda_{j}\mu_{j}|^{2} = \mathcal{O}(\delta^{2}). \end{equation} This bounds the total variation distance as $\sup_{i}\abs*{\lvert\braket{i}{\widetilde\Phi}\rvert^2-\left\lvert\braket{i}{\Phi}\right\rvert^2} = \mathcal{O}({\delta})$ and guarantees that $\|{\ket{\widetilde\Phi}-\ket{\Phi}}\| = \mathcal{O}({\delta})$. 




\subsection{\Cref{alg:find_eta}}\label{sec:alg_analysis}
\subsubsection{Accuracy of approximation}
We begin our analysis by showing that the support on $\ket{0^{s}}$ of the output state $\ket{\Psi}$ of phase estimation [cf.\ Eq.\ \eqref{eqn:Psi}] is a reasonable estimate of the amplitude $\braket{r}{\Phi} = \sin\beta$. Letting $\ket{\Psi}_2$ denote the state of the second register of $\ket{\Psi}$, $\abs{\braket{0^s}{\Psi}_{2}}$ is simply the inverse of normalisation constant in Eq.\ \eqref{Phi_tilde}. Hence, recalling Eq.\ \eqref{delta_squared}, we fix a precision $\delta$ and write $\abs{\braket{0^s}{\Psi}_{2}}^2 = \sin^2\beta + \mathcal{O}({\delta^{2}})$. For any $\widetilde\beta$, we then have
\begin{align*}
   \abs*{\sin^2\beta - \sin^2\widetilde\beta} & \leq \abs*{\abs{\braket{0^s}{\Psi}_2}^2 - \sin^2\widetilde\beta} + \abs*{\abs{\braket{0^s}{\Psi}_{2}}^2 - \sin^2\beta} \\
      &=  \abs*{\abs{\braket{0^s}{\Psi}_2}^2 - \sin^2\widetilde\beta} +  \bigO{\delta^2}.
\end{align*}
This means that if $\sin \widetilde\beta$ is an accurate estimate of $\abs{\braket{0^s}{\Psi}_2}$, then it approximates $\sin \beta$ with similar accuracy. 

Now, suppose that $\abs*{\abs*{\braket{0^s}{\Psi}_2}^2 - \sin^2\widetilde\beta} \leq {1}/{\sigma}$ for some $1/\sigma \geq \mathcal{O}({\delta^2})$. This gives
\[
    \abs*{\sin^2\beta - \sin^2\widetilde\beta} \leq \frac{1}{\sigma} + \bigO{\delta^2} = \bigO{\frac{1}{\sigma}}.
\]
We use this inequality to bound the error in $\widetilde\eta$ returned by \Cref{alg:find_eta}. Let $\eta$ be the final value queried by \Cref{alg:find_eta}, $\widetilde\eta = \eta\cot^2\widetilde\beta$ the output of amplitude estimation (step 4 of \Cref{alg:find_eta}), and $\sin^2\beta$ the amplitude that was estimated. By Eq.\ \eqref{eqn:angle} and \Cref{thm: eta_bound}, $\tan^2\beta = {\eta}/{\overline\eta}$. Note that for any $\eta$ queried by \Cref{alg:find_eta}, this results in $\overline\eta = \eta\cot^2\beta$ and $\widetilde\eta = \eta\cot^2\widetilde\beta$.
\begin{align*}
    \abs*{\overline\eta - \widetilde\eta} &= \eta \abs*{\cot^2\beta - \cot^2\widetilde\beta} \\
    &= \eta \abs*{\frac{1}{\sin^2\beta}-\frac{1}{\sin^2\widetilde\beta}} \\
    &= \frac{\eta}{\sin^2\beta \sin^2\widetilde\beta}\abs*{\sin^2\beta - \sin^2\widetilde\beta} \\
    &= \frac{\overline\eta}{\cos^2\beta\sin^2\widetilde\beta}\abs*{\sin^2\beta - \sin^2\widetilde\beta}.
\end{align*}
Since \Cref{alg:find_eta} exits when $\abs{\widetilde\beta - \pi/4} \leq \pi/16$,
\begin{align*}
    \abs*{\overline\eta - \widetilde\eta} &\leq \frac{4\overline\eta}{\cos^2\beta}\abs*{\sin^2\beta - \sin^2\widetilde\beta}\\
    &\leq 4 \overline\eta \frac{\abs*{\sin^2\beta - \sin^2\widetilde\beta}}{\abs*{\cos^2\widetilde\beta - \abs*{\sin^2\beta - \sin^2\widetilde\beta}}}\\
    &\leq 16 \overline\eta \frac{\abs*{\sin^2\beta - \sin^2\widetilde\beta}}{\abs*{1 - 4\abs*{\sin^2\beta - \sin^2\widetilde\beta}}}\\
    &\leq 16 \overline\eta \frac{1+\sigma\bigO{\delta^2}}{\abs*{\sigma - 4(1+\sigma\bigO{\delta^2})}} \\
    &= \bigO{\frac{\overline\eta}{\sigma}}.
\end{align*}
Thus, the relative error in our estimate of $\overline\eta$ can be reduced arbitrarily, up to the error of phase estimation itself, by taking amplitude estimation to greater precision.

\subsubsection{Failure rate}\label{sec:failure}

It remains to analyse the failure rate of \Cref{alg:find_eta}. Consider Step 5 of \Cref{alg:find_eta}. Let $X_i \in \{0,1\}$ be a coin that returns $0$ whenever $\abs{\widetilde\beta_i - \pi/4} \leq \pi/16$ and $1$ otherwise. By \Cref{thm:amplitude_estimation}, $\Pr{X_i = 0 \; \vert \; \abs{\beta_\eta - \pi/4} \geq \pi/8} \leq c\gamma_2$ for some constant $c$. We apply the Chernoff-Hoeffding theorem \cite{hoeffding1963probability} with probability $p = c\gamma_2$ and $\gamma_1 \log(1/\delta_0)$ repetitions of the experiment to get,
\begin{align*}
    \mathrm{Pr}\left(\frac{1}{\gamma_1 \log(1/\delta_0)}\sum_i {X_i} \geq p + \varepsilon \; \middle| \; \abs{\beta_\eta - \pi/4} \geq \pi/8 \right) &\leq \left[ \left(\frac{p}{p+\varepsilon}\right)^{p+\varepsilon} \left(\frac{1-p}{1-p-\varepsilon} \right)^{1-p -\varepsilon} \right]^{\gamma_1 \log(1/\delta_0)}\\
    &= \left[4 p (1-p)\right]^{\gamma_1 \log(1/\delta_0)/2}\\
    &\leq (4 p)^{\gamma_1 \log(1/\delta_0)/2}\\
    &= \left(4c\gamma_2\right)^{\gamma_1 \log(1/\delta_0)/2}.
\end{align*}
For an appropriate choice of $\gamma_2 < {1}/{2c}$, we have $p < {1}/{2}$ and we let let $\varepsilon = 1/2 - p$. This also yields
\[
    \Pr{\frac{1}{\gamma_1 \log(1/\delta_0)}\sum_i {X_i} \leq \frac{1}{2} \; \vert \;  \abs{\beta - \pi/4} \leq \pi/8 } \leq \left(4c\gamma_2\right)^{\gamma_1 \log(1/\delta_0)/2}.
\]
Therefore, there exists a $\gamma_2 < {1}/(8c\sqrt{n})$ such that the failure rate for each loop of \Cref{alg:find_eta} is $\mathcal{O}({({1}/{\sqrt n})^{\gamma_1\log({1}/{\delta_0})/2}})$, where we include a factor of $\sqrt{n}$ for future convenience.

By choosing an appropriate constant $\gamma_1 > 2$ in \Cref{alg:find_eta}, we can weaken this bound to $\mathcal{O}({{\delta_0}/{n}})$, which simplifies some of the computations in this subsection. Thus, the probability that \Cref{alg:find_eta} does not exit before querying an $\eta$ corresponding to a $\beta_\eta$ such that $\abs*{\beta_\eta-\pi/4} \leq \pi/8$ is $1-\mathcal{O}({{\delta_0}/{n}})$.\footnote{The bound of $\pi/8$ is weakened from the rejection bound of $\pi/16$ in \Cref{alg:find_eta} in order to account for the possibility that $\beta_\eta \approx 3\pi/16$. In such an event, the algorithm can exit ``successfully'' while still returning a $\widetilde\beta \approx \pi/4$. This bound can be tightened arbitrarily by adjusting the multiplicative step size from multiples of $2$ to multiples of $S \in (1,2)$. Alternatively, one can simply run amplitude estimation to higher precision after an initial rough estimate has been made.} Similarly, the probability that the algorithm exits once it reaches a state for which $\abs*{\beta_\eta - \pi/4} \leq \pi/8$ is $1-\mathcal{O}(({{\delta_0}/{n}})$. The probability of success is then
\[
    \left(1-\mathcal{O}\left(\frac{\delta_0}{n}\right)\right)^{\log (d \overline\eta)} \geq 1 - \mathcal{O}\left(\frac{\delta_0\log (d \overline\eta)}{n}\right),
\]
and the algorithm returns an estimate such that $\abs*{\widetilde\eta - \overline \eta} = \mathcal{O}({\Delta \overline\eta})$, where $\Delta \in [\gamma_2, 1/8)$ is the error from amplitude estimation, with probability greater than $1-\mathcal{O}({\delta_0})$. If there are no marked vertices, the algorithm runs to completion and returns $\eta = \infty$ with probability given by the above expression with $\overline\eta = n$.


For our purpose of backtracking, it does not help to prepare a state with precision $\delta$ less than $\sqrt{\Delta}$. Hence, for any iteration of \Cref{alg:find_eta} with parameter $\eta$, we can always take $2^s = \mathcal{O}({\sqrt{T \eta/\Delta^3}})$ and $\epsilon = \Theta(\sqrt{{\Delta}/({T \eta})})$. This prepares a state $\ket{\widetilde{\Phi}}$ such that $\|{\ket{\widetilde\Phi}-\ket{\Phi}}\| = \mathcal{O}(\sqrt{\Delta})$ using $\mathcal{O}(\sqrt{{T \eta}/{\Delta^3}})$ applications of $R_B$ and $R_A(\eta)$. Since the complexities of preparing the state and of estimating the relevant amplitude are additive, the expected runtime for the loop with parameter $\eta$ is $\mathcal{O}({\sqrt{T\eta/{\Delta^{3}}}} + \sqrt{n}\log(1/\delta_0)) = \mathcal{O}({\sqrt{T \eta/\Delta^{3}}}\log(1/\delta_0))$.

Now, for the full sequence of queries $(\eta_i = 2^i/d)_{i=0}^{\lceil\log (d \overline\eta) \rceil}$ made by \Cref{alg:find_eta}, we can write
\[
    \sum_{i} \sqrt{\eta_i} = \sum_{i=0}^{m}\sqrt{\frac{2^i}{d}} = \frac{2^{({m+1})/{2}}-1}{d(\sqrt{2}-1)} = \mathcal{O}\left(\sqrt{\overline\eta}\right)
\]
provided that $2^{({m+1})/{2}} \leq \sqrt{2\overline\eta}$. If the algorithm exits successfully, i.e., when $\abs{\beta_\eta - \pi/4} \leq \pi/8$, this is always the case. Because the probability of exiting with $\beta_\eta > 3\pi/8$ (which requires runtime greater than $\mathcal{O}({\sqrt{T\overline\eta}})$) is bounded by $\mathcal{O}({\delta_0\log(n)/n})$, we know that these occurrences will not contribute to the scaling of the expected runtime. Thus, \Cref{alg:find_eta} runs in $\mathcal{O}(\sqrt{T\overline\eta/\Delta^{3}}\log(1/\delta_0))$ steps and requires at most an additional $\mathcal{O}({\log \sqrt n})$ ancilla qubits. If our goal is only to determine the existence of a marked vertex, we can interpret any output of $\widetilde\eta < \infty$ as an indication that a marked vertex exists and $\widetilde\eta = \infty$ as nonexistence.  

\subsection{\Cref{alg:find_marked}}\label{sec:alg_analysis2}

\Cref{alg:find_marked} implements an extended version of Montanaro's formulation for trees containing a unique marked vertex \cite{montanaro2015quantum}. With the parameter $\eta$ in $R_{B}R_{A}(\eta)$ set to $\overline{\eta}$, the probability that the outcome of measuring $\ket{\widetilde{\Phi}}$ in the computational basis is a vertex other than the root is is within a constant of $1/2$. Conditioned on each iteration of the loop sampling a non-root vertex, the algorithm returns a marked vertex with $\mathcal{O}({\log (k \overline\eta)})$ measurements, as proven in \Cref{sec:expected_steps}. We hence fix an precision $\delta = \mathcal{O}(1/\log (k \overline\eta))$. Then, by the same argument as that in \cite{montanaro2015quantum}, letting $2^s = \Theta({\sqrt{T\eta}}/{\delta^3})$ and $\epsilon = \Theta(\delta/\sqrt{T\eta})$ guarantees that with probability $\Omega(1)$, none of the $\mathcal{O}(\log (k \overline\eta))$ measurements sample away from the support of $\ket{\Phi}$. Matching $\delta_0$ from \cref{alg:find_eta} to $\delta$, the expected runtime is $\mathcal{O}({\sqrt{T\eta}\log^3 (k \overline\eta)})$ per step, and therefore $\mathcal{O}(\sqrt{T\overline\eta_{\mathrm{max}}}\log^4 (k \overline\eta))$ for the entire procedure, where $\overline{\eta_{\mathrm{max}}}$ is defined in Eq.\ \eqref{eta_max}. The precision can be increased arbitrarily through repetition. Up to logarithmic factors, in all cases where $k \sim \poly(n)$, this results in at least a $\widetilde{\mathcal{O}}(n)$ improvement over Montanaro's algorithm of \Cref{thm:Mont_2}. For problems in which the effective resistance of all subtrees is $\mathcal{O}({\log n})$, the speedup is $\widetilde{\mathcal{O}}({n^{3/2}})$.

If an upper bound on $k$ is not provided (or we suspect that the bound is not tight), we can initially take $k=1$, and double our guess of $k$ until a marked vertex is returned. This will take at most $\log k$ repetitions. If we reach a $k$ such that $\log k\sim n$ without having found a marked vertex, we can use \Cref{alg:find_eta} in the manner of the algorithm of \Cref{thm:Mont_2}. Namely, starting at the root $r$, whenever \Cref{alg:find_eta} with input vertex $r$ returns $\eta < \infty$, we apply \Cref{alg:find_eta} to each child $c \leftarrow r$. If \Cref{alg:find_eta} returns returns $\eta < \infty$ for some $c$, then we repeat the process with $c$ as the new input. Choosing $\Delta < {1}/{8}$ and $\delta_0 = \mathcal{O}({{1}/{n}})$, we expect a success rate of $\Omega({1})$. Thus, we return a marked vertex in time $\mathcal{O}({n\sqrt{T \overline \eta}})$. Hence, for low-resistance trees with exponentially many marked vertices, this still results in a $\widetilde{\mathcal{O}}({\sqrt{n}})$ improvement over the algorithm of \Cref{thm:Mont_2}.

As we will see in the following section, the iterative approach to finding a suitable value for the parameter $k$ can sometimes result in an overall tighter bound. This suggests that there exists a strategy which can reduce logarithmic factors by using the estimates of $\overline{\eta}$ to deduce information about the structure of the tree.

\section{Optimality of the algorithm: Grover search revisited} \label{sec:grover}
In the case of unstructured search, the bound achieved in subsection \ref{sec:alg_analysis2} converges to Grover search and is therefore optimal in at least some regimes. To see this, we can consider a database where all $N$ database states are connected to a root vertex. If $k$ of the $N$ states are marked, the effective resistance of this tree is simply $\overline\eta = {1}/{k}$ and the expected steps to find a marked element scales like $\mathcal{O}(\sqrt{T/k} \log^4 k)$. Noting that $T = N+1$, the na\"ive bound is $\mathcal{O}(\sqrt{N/k} \log^4 k)$. 

The factor of $\log^4 k$ can be removed by considering the iterative approach outlined in subsection \ref{sec:alg_analysis2}. That is, we begin by setting $k=1$ and double the value of $k$ until a marked vertex is returned. Since we take only one step (from the root vertex to one of the database states, which are all leaves in the tree) and the logarithmic factors are merely there to ensure that our state is sufficiently precise to complete the necessary number of steps, we only need the sampling procedure to succeed once per marked vertex. Thus, this approach actually scales as $\mathcal{O}({\sqrt{N/k}})$, with no logarithmic factors.

\section{The effective resistance of trees and the parameter $\overline\eta$}\label{sec:math}
In this section, we prove a number of helpful lemmas concerning the particular function defined by Eqs. \eqref{E: kappas_recurrence}-\eqref{E: path_leaf} on the solution tree $\widetilde{\mathcal{T}}$ rooted at $r$, which we subsequently use to establish a connection to effective resistance in the following theorem. These results also arise naturally by identifying the harmonic function in the Rayleigh quotient of the operator $R_B R_A(\eta)$ with the minimum energy function of a relevant flow network on $\widetilde{\mathcal{T}}$, where the set of leaves $\mathcal{M}$ are all identified with a single point. Nevertheless, we construct this section in a self-contained fashion. 

Consistent with the definition of $\widetilde{\mathcal{T}}$ in subsection \ref{sec:state_features}, we denote the set of leaves in $\widetilde{\mathcal{T}}$ by $\mathcal{M}$, and the set of leaves in the subtree $\widetilde{\mathcal{T}}(v)$ rooted at $v \in \widetilde{\mathcal{T}}$ by $\mathcal{M}(v)$. (In the context of the previous sections, the leaves of $\widetilde{\mathcal{T}}$ comprise the ``shallowest marked vertices." The results derived in this section, however, are applicable to arbitrary trees $\widetilde{\mathcal{T}}$ with root $r$ and leaf set $\mathcal{M}$.)

\begin{restatable}{thm}{treeprops}\label{thm:tree_properties}
    Let $\kappa: V(\widetilde{\mathcal{T}}) \longrightarrow \mathbb{R}$ be such that
    \begin{enumerate}
        \item $\sum\limits_{\substack{v\in\widetilde{\mathcal{T}}\\v\neq r}}\kappa_{v}^{2} = 1$,
        \item $\kappa_v = \sum\limits_{c\leftarrow v} \kappa_c$ for all $v \not\in \mathcal{M}$, and 
        \item for any $v \in \widetilde{\mathcal{T}}$, $\sum\limits_{u \in \mathcal{P}(v,m_0)}\kappa_u = \sum\limits_{u \in \mathcal{P}(v,m_1)}\kappa_u$ for all leaves $m_0,m_1 \in \mathcal{M}({v})$, 
    \end{enumerate}
    where $\kappa_v \equiv \kappa(v)$. Then, $\overline\eta: \{\widetilde{\mathcal{T}}(v)\}_{v\in\widetilde{\mathcal{T}}} \longrightarrow \mathbb{R}_{\geq 0}$ defined by
    \[
        \overline\eta(v) \equiv \overline\eta(\widetilde{\mathcal{T}}(v)) = \frac{\sum\limits_{w \in \tree{v}}\kappa_w^2}{\kappa_v^2} - 1
    \]
    is an invariant of the set of subtrees $\{\widetilde{\mathcal{T}}(v)\}_{v \in \widetilde{T}}$, and $\widetilde{\eta}(v)$ equals the effective resistance between $v$ and $\mathcal{M}(v)$ (through $\widetilde{\mathcal{T}}(v)$) for each $v\in\widetilde{\mathcal{T}}$. Furthermore, if $\widetilde{\mathcal{T}}(v)$ is of maximum degree $d$ and maximum depth $n$ and has at most $k$ leaves, then \[\overline{\eta}(v) \in \left[\max\left(\frac{1}{k},\frac{1}{d}\right),n\right].\]
\end{restatable}

The remainder of this section will develop the proof of this theorem. In \Cref{lem: same_phase}, we demonstrate that for $\kappa$ satisfying the conditions of \Cref{thm:tree_properties}, $\kappa_{v_0}$ and $\kappa_{v_1}$ have the same sign for all vertices $v_0, v_1 \in\widetilde{\mathcal{T}}$. Lemmas \ref{lem: kappa_recursion} and \ref{cor: kappa_recursion} are technical lemmas which are used to prove \Cref{thm:resistance}, our first statement about effective resistance. \Cref{thm: eta_bound} completes the proof of \Cref{thm:tree_properties} by bounding the effective resistance.  For the remainder of this section, we assume that $\kappa$ satisfies the hypotheses of \Cref{thm:tree_properties}, i.e., Eqs. \eqref{E: kappas_recurrence}-\eqref{E: path_leaf}.

\begin{lemma}\label{lem: same_phase} 
For all $v_{0},v_{1} \in \widetilde{\mathcal{T}}$, $\kappa_{v_{0}}\kappa_{v_{1}} > 0$.
\begin{proof}
First, we claim that for every vertex $v \in \widetilde{\mathcal{T}}$ and any leaf $m \in \mathcal{M}(v)$, 
\[ \kappa_{m}\sum_{u \in \mathcal{P}(v,m)}\kappa_{u} > 0. \]
This is trivially true if $v$ is a leaf. Now, assume this holds true for every child $c$ of a vertex $v\not\in\mathcal{M}$. Then, for any two children $c_{0}$, $c_{1}$ of $v$ and leaves $m_{0} \in \widetilde{\mathcal{T}}(c_{0})$, $m_{1} \in \widetilde{\mathcal{T}}(c_{1})$,
\begin{align*}
    0 &< \left(\kappa_{m_{0}}\sum_{u\in\mathcal{P}(c_{0},m_{0})}\kappa_{u}\right)\left(\kappa_{m_{1}}\sum_{u\in\mathcal{P}(c_{1},m_{1})}\kappa_{u}\right) \\
    &= \kappa_{m_{0}}\kappa_{m_{1}}\left(\sum_{u\in\mathcal{P}(v,m_{0})}\kappa_{u}-\kappa_{v}\right)\left(\sum_{u\in\mathcal{P}(v,m_{1})}\kappa_{u} - \kappa_{v}\right) \\
    &= \kappa_{m_{0}}\kappa_{m_{1}}\left(\sum_{u\in\mathcal{P}(v,m_{0})}\kappa_{u}-\kappa_{v}\right)^{2},
\end{align*}
where in the last line we have used Eq.\ \eqref{E: path_leaf}.
Thus, $\kappa_{m_{0}}\kappa_{m_{1}} > 0$ for all $m_{0},m_{1} \in \mathcal{M}(v)$, from which it directly follows that
\[ \kappa_{m_{0}}\sum_{u\in\mathcal{P}(v,m_{0})}\kappa_{u} = \sum_{u\in\mathcal{P}(v,m_{0})}\sum_{m\in\mathcal{M}(u)}\kappa_{m_{0}}\kappa_{m} > 0,\]
since for $u\in\mathcal{P}(v,m_{0}) \subseteq \widetilde{\mathcal{T}}(v)$, $m \in \mathcal{M}(u)$ implies that $m\in \mathcal{M}(v)$.

Therefore, the claim is true for all $v \in \widetilde{\mathcal{T}}$, including for $r$, and so $\kappa_{m_{0}}\kappa_{m_{1}} > 0$ for all $m_{0}, m_{1} \in \mathcal{M}$. Then,
\[ \kappa_{v_0}\kappa_{v_1} = \left(\sum_{m_0\in \mathcal{M}(v_0)}\kappa_{m_0}\right)\left(\sum_{m_1\in\mathcal{M}(v_1)}\kappa_{m_1}\right) = \sum_{m_0\in\mathcal{M}(v_0)}\sum_{m_1\in\mathcal{M}(v_1)}\kappa_{m_0}\kappa_{m_1} > 0,\]
as desired.
\end{proof}
\end{lemma}

Without loss of generality, we henceforth assume that $\kappa_{v} > 0$ for all $v \in \widetilde{\mathcal{T}}$.

\begin{lemma} \label{lem: kappa_recursion} For every $v \in \widetilde{\mathcal{T}}$ and any $m \in \mathcal{M}(v)$,
\[ \kappa_{v}\sum_{u\in\mathcal{P}(v,m)}\kappa_{u} = \sum_{u\in\widetilde{\mathcal{T}}(v)}\kappa_{u}^{2}. \]

\begin{proof}
This is trivially true if $v \in \mathcal{M}$, and assuming that this holds for every child $c$ of $v \not\in\mathcal{M}$, we have for any $m \in\mathcal{M}(v)$
\begin{align*}
    \kappa_{v}\sum_{u\in\mathcal{P}(v,m)}\kappa_{u}
    &= \kappa_{v}\left(\kappa_{v} + \sum_{\substack{u\in\mathcal{P}(v,m)\\u\neq v}}\kappa_{u}\right) \\
    &= \kappa_{v}^{2} + \sum_{c \leftarrow v}\kappa_{c}\sum_{u\in\mathcal{P}(c,m)}\kappa_{u} \\
    &= \kappa_{v}^{2} + \sum_{c\leftarrow v}\sum_{u\in\widetilde{\mathcal{T}}(c)}\kappa_{u}^{2} \\
    &= \sum_{u\in\widetilde{\mathcal{T}}(v)}\kappa_{u}^{2},
\end{align*}
where the second line follows from Eq.\ \eqref{E: kappas_recurrence}, the third from Eq.\ \eqref{E: path_leaf}, and the fourth from the induction hypothesis.
\end{proof}
\end{lemma}

\begin{lemma} \label{cor: kappa_recursion}
For any $m \in \mathcal{M}$,
\[ \kappa_{r}\sum_{\substack{u\in\mathcal{P}(r,m)\\u\neq r}}\kappa_{u} = 1. \]

\begin{proof}
This follows directly from \Cref{lem: kappa_recursion} and Eq.\ \eqref{E: kappas_norm}.
\end{proof}
\end{lemma}

Belovs' original result relates the commute time of a classical random walk with the problem of determining connectivity via a quantum walk. In particular, the commute time can be characterised in terms of the effective resistance of the corresponding graph, provided that the graph's structure is known \textit{a priori} \cite{belovs2013quantum}. In light of this, the following theorem, which is essentially a statement about the effective resistance of trees, should not be too surprising.

\begin{thm}\label{thm:resistance}
    For every vertex $v\in \widetilde{\mathcal{T}}$,
    \begin{equation} \label{def: etabar} \overline{\eta}(v) \equiv \frac{\sum\limits_{u \in \widetilde{\mathcal{T}}(v)}\kappa_{u}^{2}}{\kappa_v^2} - 1\end{equation}
    is a unique map for all $\{\widetilde{\mathcal{T}}(v)\}_{v\in\widetilde{\mathcal{T}}}$. 
    
    \begin{proof}
    We begin by noting that for a leaf $m\in\mathcal{M}$,
    \[\overline{\eta}(m) = \frac{\sum\limits_{u\in\widetilde{\mathcal{T}}(m)}\kappa_{u}^{2}}{\kappa_m^2} -1 = 0\]
    is clearly unique for the trivial graph $\widetilde{\mathcal{T}}(m)$.
    Now, assume that for every child $c$ of $v \not\in\mathcal{M}$, $\overline{\eta}(c) = \sum\limits_{u \in \widetilde{\mathcal{T}}(c)}\kappa_{u}^{2}/\kappa_{c}^{2} - 1$ is uniquely defined for $\widetilde{\mathcal{T}}(c)$. Then,
    \begin{align*}
        \overline{\eta}(v) = \frac{\sum\limits_{u\in\widetilde{\mathcal{T}}(v)}\kappa_{u}^{2}}{\kappa_{v}^{2}} - 1 &= \frac{\kappa_{v}^{2}}{\kappa_{v}^{2}} + \frac{\sum\limits_{c\leftarrow v}\sum\limits_{u\in \widetilde{\mathcal{T}}(c)}\kappa_{u}^{2}}{\left(\sum\limits_{c\leftarrow v}\kappa_{c}\right)^{2}} - 1 \\
        &= \frac{\sum\limits_{c\leftarrow v}\kappa_c\sum\limits_{u \in \mathcal{P}(c,m)}\kappa_u}{\left(\sum\limits_{c\leftarrow v}\kappa_c\right)^2} 
        \\
        &= \frac{\sum\limits_{u \in \mathcal{P}(c_0,m_{0})}\kappa_u}{\sum\limits_{c\leftarrow v}\kappa_c} \\
        &= \left(\sum\limits_{c\leftarrow v} \frac{\kappa_c^2}{\kappa_{c}\sum\limits_{u\in \mathcal{P}(c,m)}\kappa_u}\right)^{-1}\\ 
        &= \left(\sum\limits_{c\leftarrow v} \frac{\kappa_c^2}{\sum\limits_{u\in \tree{c}}\kappa_u^2}\right)^{-1} \\
        &= \left(\sum_{c\leftarrow v}\frac{1}{\overline{\eta}(c) + 1}\right)^{-1}
    \end{align*}
    where the second and fifth lines follow from \Cref{lem: kappa_recursion}, and the third and fourth lines use the direct consequence of Eq.\ \eqref{E: path_leaf} that for children $c$ of $v$ and leaves $m \in \mathcal{M}(c)$, $\sum_{u\in\mathcal{P}(c,m)}\kappa_{u}$ is independent of $c$.

    Therefore, for every $v \in \widetilde{\mathcal{T}}$, $\overline{\eta}(v)$ is unique and well-defined. Furthermore, since we have
    \[ \frac{1}{\overline{\eta}(v)} = \sum_{c\leftarrow v}\frac{1}{\overline{\eta}(c) + 1}, \]
    $\overline{\eta}(v)$ satisfies the recurrence relation for the effective resistance of a circuit with source $v$ and sinks $\mathcal{M}(v)$, where each edge has resistance 1. 
    \end{proof}
\end{thm}

\begin{cor}\label{cor:ratio}
For any child $c$ of $v$,
\[
    \frac{\overline\eta(c)+1}{\overline\eta(v)} = \frac{\kappa_v}{\kappa_c}.
\]
\end{cor}
\noindent
    This follows from the definition of $\overline{\eta}(v)$ in \Cref{thm:resistance} [cf.\ Eq.\ \eqref{def: etabar}] and the proof is omitted.

Since $\{\overline\eta(v)\}_{v\in\widetilde{\mathcal{T}}}$ is uniquely determined, \Cref{cor:ratio} guarantees that, up to a global factor, $\kappa$ is unique and well-defined. The global factor is fixed by the normalisation condition [Eq.\ \eqref{E: kappas_norm}].

The fact that $\overline{\eta} = \overline{\eta}(r)$ is the effective resistance of the tree $\widetilde{\mathcal{T}}$ makes the following theorem fairly intuitive.

\begin{thm}\label{thm: eta_bound}
    Let $n$ denote the depth of the tree $\widetilde{\mathcal{T}}$, $k \equiv |\mathcal{M}|$ the number of leaves, and $d_r$ the degree of the root vertex $r$. Then
    \[ \overline{\eta} \in \left[\max\left(\frac{1}{k},\frac{1}{d_r}\right), n\right].\]
    \begin{proof}
    First, we note that 
    \[ \overline{\eta} = \frac{1}{\kappa_{r}^{2}}\sum_{u\in\widetilde{\mathcal{T}}}\kappa_{u}^{2} - 1 = \frac{1}{\kappa_{r}^{2}}\sum_{\substack{u\in\widetilde{\mathcal{T}}\\u\neq r}}\kappa_{u}^{2} = \frac{1}{\kappa_{r}^{2}}. \]
    by definition.
    
    We use the Cauchy-Shwarz inequality to obtain the upper and lower bounds. By \Cref{cor: kappa_recursion}, for any $m \in \mathcal{M}$,
    \[ \frac{1}{\kappa_r^2} = \left(\sum_{\substack{u \in \mathcal{P}(r,m)\\u\neq r}}\kappa_{u}\right)^{2} \leq \left(\sum_{\substack{u \in \mathcal{P}(r,m)\\u\neq r}}1\right)\left(\sum_{\substack{u \in \mathcal{P}(r,m)\\u\neq r}}\kappa_{u}^{2}\right) \leq n\sum_{\substack{u \in \widetilde{\mathcal{T}}\\u\neq r}}\kappa_{u}^{2} = n.\]
    
    In the other direction,
    \[ \frac{1}{\kappa_{r}^{2}} = \left[\sum\limits_{m\in\mathcal{M}}\kappa_{m}\right]^{-2} \geq \left[\left(\sum\limits_{m\in\mathcal{M}}1\right)\left(\sum\limits_{m\in\mathcal{M}}\kappa_{m}^{2}\right)\right]^{-1} \geq \left(k\sum\limits_{\substack{u \in \widetilde{\mathcal{T}}\\u\neq r}}\kappa_{u}^{2}\right)^{-1} = \frac{1}{k}. \]
    In addition, as a consequence of \Cref{thm:resistance} (or simply the definition of effective resistance),
    \begin{align*}
        \frac{1}{\overline\eta} &= \sum_{c\leftarrow r} \frac{1}{\overline\eta(c) + 1} \leq \sum_{c\leftarrow r}1 = d_r.
    \end{align*}
    Thus, we have that $\overline\eta \geq \max\left(1/k, {1}/{d_r}\right)$, as claimed.
    
    \end{proof}
\end{thm}

\section{Future directions} \label{conclusion}
This paper provides some natural directions for future work. Given that we are able to achieve a bound that scales like $\mathcal{O}(\sqrt{T\overline\eta_{\mathrm{max}}})$ in the case of trees of unknown structure (and that our argument is compatible with the generalisation to directed acyclic graphs in Ref.\ \cite{ambainis2017quantum}), it would be interesting to see to what extent these algorithms can be generalised to graphs of unknown structure. The results of \Cref{sec:math}
are derived in terms of local neighbors on the appropriate state space, and it seems plausible that similar arguments from graph theory might offer insight beyond trees. We expect that as long as appropriate diffusion operators can be implemented, many of the statements above should hold for arbitrary graphs of bounded degree. The challenge, then, lies in determining these diffusion operators and deriving a theorem like the one in \Cref{sec:expected_steps}.

Aside from Montanaro's analogue of backtracking itself, we do not make use of any of methods utilised by modern (classical) backtracking algorithms, such as those considered in Ref.\ \cite{ambainis2017quantum}. The effective resistance $\overline\eta$ reveals some information about the structure of the graph, thereby giving an indication of how ``easy" it is to find a solution. One might imagine, for instance, systematically marking subsets of vertices in a tree in an attempt to elucidate its structure, without incurring much additional overhead. Although it seems clear that the information contained in estimates of $\overline\eta$ for trees or subtrees can be exploited, explicit methods for doing so remain to be developed, and the potential advantages are not yet obvious. 

Finally, it may be worthwhile to determine the exact relation between the distribution of marked vertices and the effective resistance. This could perhaps be used to derive of a tighter bound than that of \Cref{sec:expected_steps} and reduce the dependence upon $k$ of the logarithmic factors in \Cref{thm: k_marked}.

\section{Acknowledgements}
The authors would like to thank Courtney Brell, Daniel Gottesman, and Ted Yoder for useful discussions. The bound on the expected number of steps in the walk is primarily due to the MathOverflow user fedja \cite{mathoverflow}. KW is grateful for support from the Mike Lazaridis Scholarship in Theoretical Physics. \PIRA

\newpage
\appendix

\section{Proof of \Cref{lem: P_epsilon}} \label{xi_proof}

In this appendix, we use the effective spectral gap lemma to prove that 
\begin{equation}
    \|P_{\epsilon}\ket{\Phi^{\perp}}\|^{2} = \bigO{\epsilon^{2}T\eta},
\end{equation}
where $P_{\epsilon}$ is the projector onto the span of eigenvectors of $R_{B}R_{A}(\eta)$ with eigenvalues $e^{2i\theta}$ such that $|\theta| \leq \epsilon$, and $\ket{\Phi^{\perp}}$ is defined as in Eq.\ \eqref{def: Phi_perp}.

\begin{lemma}[effective spectral gap lemma \cite{lee2011quantum}]\label{lem: spectral_gap}

Let $\Pi_{A}$ and $\Pi_{B}$ be projectors, and let $R_{A} = 2\Pi_{A} - I$ and $R_{B} = 2\Pi_{B} - I$ be the reflections about their respective ranges. Let $P_{\epsilon}$ denote the projector onto $\mathrm{span}(\{\ket{\psi}:R_{B}R_{A}\ket{\psi} = e^{2i\theta}\ket{\psi},|\theta| \leq \epsilon\})$. Then, if $\Pi_{A}\ket{\xi} = 0$, $\|P_{\epsilon}\Pi_{B}\ket{\xi}\| \leq \epsilon\norm{\ket{\xi}}$.


\end{lemma}

If we are able to find vector $\ket{\xi}$ such that $\Pi_{A}\ket{\xi} = 0$ and $\Pi_{B}\ket{\xi} = \ket{\Phi^{\perp}}$, where $\Pi_A$ and $\Pi_B$ correspond to the diffusion operators $R_A$ and $R_B$ defined in subsection \ref{sec: diffusion_operators} (we will suppress the dependence on the input parameter $\eta$ in the following), \cref{lem: spectral_gap} can be applied to prove our claim. 

\begin{prop} \label{xi_facts}
Let $\Pi_{A}$ and $\Pi_{B}$ be the projectors onto $\mathrm{span}(\{\ket{\psi}:R_{A}\ket{\psi} = \ket{\psi}\})$ and $\mathrm{span}(\{\ket{\psi}:R_{B}\ket{\psi} = \ket{\psi}\})$, respectively. For $v \in\mathcal{T}(v) \setminus\{r\}$, let $p(v)$ denote the parent of $v$ in $\mathcal{T}$.  A vector $\sum\limits_{v\in \mathcal{T}}\alpha_{v}\ket{v}$ satisfies \begin{equation} \label{E: xi_conditions} 
\Pi_{A}\ket{\xi} = 0 \qquad \text{and} \qquad 
\Pi_{B}\ket{\xi} = \ket{\Phi^{\perp}} \end{equation}
if and only if
\begin{enumerate}[label=(\alph*)]
    \item \label{xi_facts_B} for all $v \in B$,
    \[ \alpha_{v} = \begin{cases} \sqrt{\eta}\alpha_{r} \qquad &\text{if $p(v) = r$} \\ \alpha_{p(v)} \qquad &\text{otherwise}, \end{cases} \]

    \item for all $v \neq r \in A$, \label{xi_facts_A}
    \begin{equation*} \alpha_{v} = \alpha_{p(v)} - \left(\kappa_{p(v)} + \kappa_{v}\right)\sin\beta,
    \end{equation*}
    
    \item \label{xi_facts_root} $\alpha_{r} = \cos\beta$, and
    
    \item \label{xi_facts_marked} for every marked vertex $m \in \mathcal{M}$,
    \[ \alpha_{m} = \begin{cases} 0 \qquad &\text{if $m \in A$} \\ 
    \kappa_{m}\sin\beta \qquad  &\text{if $m \in B$}. \end{cases} \] 
\end{enumerate}

\begin{proof}
    The eigenvalue-1 eigenspace of $R_{A}$ is spanned by the set of vectors
    \[ S_{A} \equiv \{\ket{\psi_{v}^{\perp}}:\bra{\psi_{v}^{\perp}}\psi_{v}\rangle = 0, v \in A\} \cup \{\ket{\psi_{m}}:m \in A\cap \mathcal{M}\} \]
    [cf.\ Eq.\ \eqref{def: psi_v}]. Its orthogonal complement is then spanned by $\{\ket{\psi_{v}}:v \in A\setminus \mathcal{M}\}$, so $\Pi_{A}\ket{\xi} = 0$ if and only if $\ket{\xi}$ is a linear combination of $\ket{\psi_{v}}$ corresponding to vertices $v \in A$ that are not marked. This implies that $\alpha_{m} = 0$ for all marked vertices $m \in A \cap\mathcal{M}$ (first case of item \ref{xi_facts_marked}), and item \ref{xi_facts_B} follows from the definition of $\ket{\psi_{v}}$ in Eq.\ \eqref{def: psi_v}.
    
    Similarly, the eigenvalue-1 eigenspace of $R_{B}$ is spanned by
    \[ S_{B} \equiv \{\ket{r}\} \cup \{\ket{\psi_{v}^{\perp}}:\bra{\psi_{v}^{\perp}}\psi_{v}\rangle = 0, v \in B\} \cup\{\ket{\psi_m}: m \in B \cap \mathcal{M}\}. \]
    Then, $\Pi_{B}\ket{\xi} = \ket{\Phi^{\perp}}$ is satisfied if and only if, for all $\ket{\zeta} \in S_{B}$, $\bra{\zeta}\xi\rangle = \bra{\zeta}\Phi^{\perp}\rangle$. \Cref{xi_facts_root} and the second case in item \ref{xi_facts_marked} respectively impose $\bra{r}\xi\rangle = \bra{r}\Phi^{\perp}\rangle$ and $\bra{m}\xi\rangle = \braket{m}{\Phi^{\perp}}$ for $m \in B \cap \mathcal{M}$. Moreover, for each vertex $v \in B$, the set of vectors
    \[ \left\{\ket{\psi_{v,c_i}^{\perp}} \equiv - \ket{v} + d_{v}\ket{c_i} - \sum_{c_j \leftarrow v}\ket{c_j} \right\}_{c_i \leftarrow v} \] forms a basis for the space $\mathrm{span}(\{\ket{\psi_{v}^{\perp}}:\bra{\psi_{v}^{\perp}}\psi_{v}\rangle = 0\})$. Direct calculation yields, for each child $c_{i}$,
    \[ \bra{\psi_{v,c_i}^{\perp}}\xi\rangle = -\alpha_{v} + d_{v}\alpha_{c_{i}} - \sum_{c_j \leftarrow v}\alpha_{c_j} \]
    and [cf.\ Eq.\ \eqref{def: Phi_perp}]
    \[ \bra{\psi_{v,c_i}^{\perp}}\Phi^{\perp}\rangle = -d_{v}\kappa_{c_{i}}\sin\beta.\]
    For each $v \in B$, we thus obtain a system of $d_{v} - 1$ linear equations of the form
    \[ -\alpha_{v} + d_{v}\alpha_{c_{i}} - \sum_{\substack{c_j \leftarrow v}}\alpha_{c_j} = -d_{v} \kappa_{c_{i}}\sin\beta\]
    which has the unique solution [cf.\ Eq.\ \eqref{E: kappas_recurrence}]
    \[ \alpha_{c_{i}} = \alpha_{v} - (\kappa_{v} + \kappa_{c_{i}})\sin\beta.\]
    This is \Cref{xi_facts_A}.
\end{proof}

\end{prop}

One might observe that Items (a) through (d) of \Cref{lem: P_epsilon} may over-constrain the set of coefficients $\{\alpha_{v}\}_{v\in\mathcal{T}}$. Consider, for instance, the path $\mathcal{P}(r,m)$ from the root $r$ to some marked vertex $m \in \mathcal{M}$. Items \ref{xi_facts_root} and \ref{xi_facts_marked} fix boundary conditions at $r$ and at the marked vertex $m$, respectively, while Items \ref{xi_facts_B} and \ref{xi_facts_A} constitute a recurrence relation that determines $\alpha_{v}$ for each vertex $v$ along $\mathcal{P}(r,m)$ as a function of $\alpha_{p(v)}$ for its parent $p(v)$ and the values assigned by $\kappa$. If we start with $\alpha_{r}$, given by item \ref{xi_facts_root} and recurse down the path using \ref{xi_facts_B} and \ref{xi_facts_A} to obtain $\alpha_{m}$, the result may be inconsistent with the constraint on $\alpha_{m}$ imposed by \ref{xi_facts_marked}. However, for $\kappa$ given by Eqs.\ \eqref{E: kappas_recurrence}-\eqref{E: path_leaf}, there indeed exists a set of coefficients $\{\alpha_{v}\}_{v\in\mathcal{T}}$ that satisfies the recurrence relation as well as both boundary conditions. We determine these coefficients explicitly in the following lemma.

\begin{lemma} \label{lem: alphas}
For any nontrivial input tree $\mathcal{T}$ and any parameter $\eta$, if $\kappa$ satisfies Eqs.\ \eqref{E: kappas_recurrence}-\eqref{E: path_leaf}, then there exists a unique vector $\ket{\xi} = \sum_{v\in \mathcal{T}}\alpha_{v}\ket{v}$ for which Eq.\ \eqref{E: xi_conditions} holds. Moreover, 
\begin{equation} \label{E: xi_bound} \|\ket{\xi}\| = \bigO{\sqrt{T\eta}}. \end{equation}

\begin{proof}
    First, consider any two vertices $v_{0} \in B$ and $v \in \mathcal{T}(v_{0})$. Summing over all edges $(u,w) \in \mathcal{P}(v_{0},v)$,
    \begin{align*}
        \alpha_{v_{0}} - \alpha_{v} &= \sum_{(u,w) \in \mathcal{P}(v_{0},v)}(\alpha_{u} - \alpha_{w}) \\
        &= \sum_{\substack{(u,w)\in\mathcal{P}(v_{0},v)\\u\in A}}(\alpha_{u} - \alpha_{w}) + \sum_{\substack{(u,w)\in\mathcal{P}(u,w)\\u\in B}}(\alpha_{u} - \alpha_{w}) \\
        &= \sum_{\substack{(u,w)\in\mathcal{P}(v_{0},v)\\u\in A}}0 + \sum_{\substack{(u,w)\in\mathcal{P}(v_{0},v)\\u\in B}}(\kappa_{u} + \kappa_{w})\sin\beta \\
        &= \sin\beta\begin{dcases} \sum_{u\in\mathcal{P}(v_{0},v)}\kappa_{u} \qquad &\text{if $v \in A$} \\
        \sum_{u\in\mathcal{P}(v_{0},v)}\kappa_{u} - \kappa_{v} \qquad &\text{if $v \in B$}, \end{dcases}
    \end{align*}
    where in the third line, we apply item \ref{xi_facts_B} of \Cref{xi_facts} in the first sum and item \ref{xi_facts_A} in the second. In particular, if $v_{0}$ is the ancestor of $v$ of depth 1, we have $\alpha_{v_{0}} = \sqrt{\eta}\alpha_{r} = \sqrt{\eta}\cos\beta$ by items \ref{xi_facts_B} and \ref{xi_facts_root}, and the above relation can be rearranged to obtain
    \begin{equation} \label{E: alphas} \alpha_{v} = \sin\beta\begin{dcases} 
    \frac{1}{\kappa_{r}} - \sum_{\substack{u\in\mathcal{P}(r,v)\\u\neq r}}\kappa_{u} \qquad &\text{if $v \in A$} \\
    \frac{1}{\kappa_{r}} - \sum_{\substack{u\in\mathcal{P}(r,v)\\u\neq r}}\kappa_{u} + \kappa_{v} \qquad &\text{if $v \in B$}, \end{dcases} \end{equation}
    for all $v \in \mathcal{T} \setminus \{r\}$.
    
    Taking $v$ to be a marked vertex $m \in \mathcal{M}$, this gives us 
    \begin{align*} \alpha_{m} &= \sin\beta\begin{dcases} \frac{1}{\kappa_{r}} - \sum_{\substack{u\in\mathcal{P}(r,m)\\u\neq r}}\kappa_{u} \qquad &\text{if $v \in A$} \\ \frac{1}{\kappa_{r}} - \sum_{\substack{u\in\mathcal{P}(r,m)\\u\neq r}}\kappa_{u} + \kappa_{m} \qquad &\text{if $v \in B$} \end{dcases} \\
    &= \begin{dcases} 0 \qquad &\text{if $v \in A$} \\ \kappa_{m}\sin\beta \qquad &\text{if $v \in B$}, \end{dcases}
    \end{align*}
    as required by Eq.\ \eqref{E: xi_conditions}. This uses \Cref{cor: kappa_recursion}, which we note is a consequence of the constraints on $\kappa$ imposed by Eq.\ \eqref{E: path_leaf}. Thus, $\ket{\xi} = \sum_{v\in\mathcal{T}}\alpha_{v}\ket{v}$ satisfies Eq.\ \eqref{E: xi_conditions} and is uniquely determined. 
    
    It follows from Lemmas \ref{lem: same_phase} and \ref{cor: kappa_recursion} that for all $v \in \mathcal{T}\setminus \{r\}$,
    \[ \left|\sum_{\substack{u\in\mathcal{P}(r,m)\\u\neq r}}\kappa_{u} - \kappa_{v} \right| \leq \left|\sum_{\substack{u\in\mathcal{P}(r,m)\\u\neq r}}\kappa_{u} \right| \leq \left|\frac{1}{\kappa_{r}}\right|, \]
    so $|\alpha_{v}| \leq |\sin\beta/\kappa_{r}| = \sqrt{\eta}|\cos\beta|$ by Eq.\ \eqref{E: alphas} and \Cref{lem: same_phase}. Therefore, 
    \begin{align*} \|\ket{\xi}\|^{2} &= \alpha_{r}^{2} + \sum_{\substack{v\in\mathcal{T}\\v\neq r}}\alpha_{v}^{2}
    \leq \cos^{2}\beta + \sum_{\substack{v\in\mathcal{T}\\v\neq r}}\eta\cos^{2}\beta
    = \left(1 + (T-1)\eta\right)\cos^{2}\beta
    \leq 2 (T-1)\eta \cos^2 \beta,
    \end{align*}
    where the last inequality follows from the fact that $\eta \geq 1/k_{\mathrm{max}} \geq 1/(T-1)$ for any nontrivial tree whose root vertex is not marked.
\end{proof}
\end{lemma}
Finally, the lemmas above imply the following result, which is applied in subsection \ref{sec: prepare_phi} to bound the complexity the phase estimation procedure used to prepare the state $\ket{\Phi}$.

\Pepsilon*
    \begin{proof}
        Let $\ket{\xi}$ be the vector satisfying Eq.\ \eqref{E: xi_conditions}. Then,  
        \[ \|P_{\epsilon}\ket{\Phi^{\perp}}\| = \|P_{\epsilon}\Pi_{B}\ket{\xi}\| \leq \epsilon\|\ket{\xi}\| = \bigO{\epsilon\sqrt{T\eta}}, \]
        where the inequality follows from \Cref{lem: spectral_gap} and the last equality is Eq.\ \eqref{E: xi_bound}.
    \end{proof}

\section{Expected number of steps to reach a marked vertex}\label{sec:expected_steps}
 
We reproduce a theorem originally derived in Ref.\ \cite{mathoverflow}, altered slightly so as to make use of our existing theorems and notation.

\begin{thm}\label{thm:fedja}
    Let $(X_i \in \widetilde{\mathcal{T}})_{i=0}^{n}$ be  a discrete-time Markov chain on the state space $V(\widetilde{\mathcal{T}})$, the vertex set of the tree $\widetilde{\mathcal{T}}$ rooted at $r$ with effective resistance $\overline\eta$, and leaves $\mathcal{M}$. Let $X_0 = r$ and suppose that for $u,v \in \widetilde{\mathcal{T}}$,
    \[
        \Pr{X_{i+1} = u \vert X_i = v} =
            \begin{dcases}
                \frac{\kappa_u^{2}}{ 
           \sum\limits_{w\in\widetilde{\mathcal{T}}(v)\setminus\{v\}}\kappa_{w}^{2}} \qquad &\text{if $u \in \widetilde{\mathcal{T}}(v)\setminus\{v\}$} \\
                0 \qquad &\text{otherwise,}
            \end{dcases}
    \]
    where $\widetilde{\mathcal{T}}(v)$ denotes the subtree rooted at $v$ and $\kappa$ satisfies the hypotheses of \cref{thm:tree_properties}.
    Then, the expected amount of time $\E{r}$ until $X_i \in \mathcal{M}$ is bounded as
    \[
        \E{r} \leq \log\left[|\mathcal{M}|(\overline{\eta} + 1)\right].
    \]
\end{thm}
\begin{proof}
For $v \in \widetilde{\mathcal{T}}$, we denote the effective resistance of the subtree $\widetilde{\mathcal{T}}(v)$ by $\overline{\eta}(v)$ and the set of leaves in $\widetilde{\mathcal{T}}(v)$ by $\mathcal{M}(v)$. 

By Theorem 10, we can write
\begin{align*}
     \Pr{X_{i+1} = u \vert X_i = v} = \frac{\kappa_u^2}{\kappa_v^2 \overline\eta(v)}.
\end{align*}\newpage
Let $\E{v}$ denote the expected time to reach a leaf starting from $v \not\in\mathcal{M}$. Then,
\begin{align*}
    \E{v} &= 1 + \sum_{\substack{u\in\widetilde{\mathcal{T}}(v)\\u\neq v}}\Pr{X_{i+1} = u \vert X_i = v}\E{u} \\
    &= 1 + \sum_{c\leftarrow v}\sum_{u\in\widetilde{\mathcal{T}}(c)}\frac{\kappa_u^2}{\kappa_v^2 \overline\eta(v)}\E{u} \\
    &= 1 + \sum_{c\leftarrow v}\frac{1}{\kappa_v^2\overline\eta(v)}\left[ \kappa_c^2\E{c} + \sum_{\substack{u \in \widetilde{\mathcal{T}}(c) \\ u \neq c}}\kappa_u^2\E{u} \right] \\
    &= 1+\sum_{c\leftarrow v}\frac{\kappa_c^2}{\kappa_v^2 \overline\eta(v)}\left[\E{c} + \overline\eta(c)(\E{c}-1)\right] \\
    &= \sum_{c\leftarrow v}\frac{\kappa_c}{\kappa_v}\left[1 - \frac{\kappa_{c}\overline{\eta}(c)}{\kappa_{v}\overline{\eta}(v)} + \frac{\kappa_{c}(1 + \overline{\eta}(c))}{\kappa_{v}\overline{\eta}(v)}\E{c}\right] \\
    &= \sum_{c\leftarrow v}\frac{\kappa_c}{\kappa_v}\left[\frac{1}{\overline\eta(c) + 1}+\E{c}\right]
\end{align*}
where the last line follows from \cref{cor:ratio}.

Now, we claim that for all $v$, 
\[ \E{v} \leq \sum_{m\in\mathcal{M}(v)}\frac{\kappa_{m}}{\kappa_{v}}\log\left(\kappa_v\frac{\overline{\eta}(v)+1}{\kappa_{m}}\right). \]
This is trivially true if $v \in \mathcal{M}$. Assume that it holds for all children of $v\not\in\mathcal{M}$. Then,
\begin{align*}
    \E{v} &\leq \sum_{c\leftarrow v}\frac{\kappa_{c}}{\kappa_{v}}\left[\frac{1}{\overline\eta(c)+1} + \sum_{m\in\mathcal{M}(c)}\frac{\kappa_{m}}{\kappa_{c}}\log\left(\kappa_c\frac{\overline{\eta}(c)+1}{\kappa_{m}}\right)\right] \\
    &=\sum_{c\leftarrow v}\frac{\kappa_c}{\kappa_v}\sum_{m\in\mathcal{M}(c)}\frac{\kappa_m}{\kappa_c}\left[\frac{\kappa_c}{\kappa_c(\overline\eta(c)+1)} + \log\left(\kappa_c\frac{\overline{\eta}(c) + 1}{\kappa_{m}}\right)\right] \\
    &\leq\sum_{c\leftarrow v}\sum_{m\in\mathcal{M}(c)}\frac{\kappa_m}{\kappa_v}\left[\log\left(\kappa_c\frac{\overline{\eta}(c) + 2}{\kappa_{m}}\right)\right] \\
    &=\sum_{c\leftarrow v}\sum_{m\in\mathcal{M}(c)}\frac{\kappa_m}{\kappa_v}\left[\log\left(\frac{\kappa_v\overline\eta(v)+\kappa_c}{\kappa_{m}}\right)\right] \\
    &\leq \sum_{m \in \mathcal{M}(v)}\frac{\kappa_{m}}{\kappa_{v}}\left[\log \left(\kappa_v\frac{\overline{\eta}(v) + 1}{\kappa_m}\right)\right],
\end{align*}
where the fourth line follows from \cref{cor:ratio} and the final line from \cref{lem: same_phase} and \cref{E: kappas_recurrence}.
Restricting to the case of $v = r$ and applying Jensen's inequality, we find that
\[
    \E{r} \leq \log \left(\abs{\mathcal{M}}\right) + \log\left( \overline\eta +1 \right),
\]
as desired.
\end{proof}

\end{document}